%% file: mingreedy_priorityalgos.tex
\documentclass[11pt]{article}
\usepackage[latin1]{inputenc}
\usepackage[T1]{fontenc}

\usepackage{lmodern}				
\usepackage{fullpage}
\usepackage{microtype}
\usepackage{hyperref}				
\usepackage{xspace}

\usepackage[cmex10]{amsmath}
\usepackage{amssymb,amsthm,url}
\usepackage{complexity}				
\usepackage[noend]{algpseudocode} 	

\usepackage[normalem]{ulem}
\usepackage[inline]{enumitem}

\usepackage{tikz} 					
\usepackage{float}
\usetikzlibrary{positioning}
\usetikzlibrary{calc}
\usetikzlibrary{arrows}
\usetikzlibrary{decorations.markings}
\usetikzlibrary{decorations.pathmorphing}

\newtheorem{theorem}{Theorem}
\newtheorem{definition}{Definition}

\newtheorem{lemma}{Lemma}

\newcommand\AlgFrieze{\textsc{MRG}\xspace}
\newcommand\AlgKarp{\textsc{Ranking}\xspace}
\newcommand\MDeg{\textsc{MinGreedy}\xspace}
\newcommand\EDSM{\textsc{eDSM}\xspace}
\newcommand\MDS{\textsc{MDS}\xspace}
\newcommand{\maxdeg}{{\ensuremath{\Delta}}\xspace}
\newcommand{\mopt}{\ensuremath{{M^*}}\xspace}
\newcommand{\m}{\ensuremath{M}\xspace}
\newcommand{\oneOverTwo}{\ensuremath{{\frac{1}{2}}}\xspace}

\newcommand{\twoOverThree}{\ensuremath{{\frac{2}{3}}}\xspace}
\newcommand{\desiredRatio}{\ensuremath{{\frac{\maxdeg-1}{2\maxdeg-3}}}\xspace}
\newcommand{\weakerRatio}{\ensuremath{{\frac{\maxdeg-1/2}{2\maxdeg-2}}}\xspace}
\newcommand{\transferred}{\ensuremath{\theta}\xspace}
\newcommand{\mingreedy}{\textsc{MinGreedy}\xspace}
\newcommand{\mrg}{\textsc{MRG}\xspace}
\newcommand{\karpsipser}{\textsc{KarpSipser}\xspace}

\newcommand{\greedy}{\textsc{Greedy}\xspace}
\newcommand{\ranking}{\textsc{Ranking}\xspace}

\newcommand{\g}{{\ensuremath{G}}\xspace}
\newcommand{\mg}{{\ensuremath{H}}\xspace}
\newcommand{\w}[1]{{\ensuremath{m_{#1}}}\xspace}
\newcommand{\wopt}[1]{{\ensuremath{m_{#1}^*}}\xspace}
\newcommand{\mgapprox}{{\ensuremath{\alpha}}\xspace}

\newcommand{\edge}[1]{\ensuremath{(#1)}\xspace}

\newcommand{\sdown}{\ensuremath{s_{\downarrow}}\xspace}
\newcommand{\udown}{\ensuremath{u_{\downarrow}}\xspace}
\newcommand{\vdown}{\ensuremath{v_{\downarrow}}\xspace}
\newcommand{\inlineheading}[1]{\textbf{#1}}

\newcommand{\oneonepath}{singleton\xspace}
\newcommand{\oneonepaths}{singletons\xspace}
\newcommand{\onetwopath}{\oneOverTwo-path\xspace}
\newcommand{\onetwopaths}{\oneOverTwo-paths\xspace}
\newcommand{\mmoptpath}{augmenting path\xspace}
\newcommand{\mmoptpaths}{augmenting paths\xspace}

\newcommand{\oneonePaths}{Singletons\xspace}

\newcommand{\mmoptPaths}{Augmenting Paths\xspace}

\begin{document}

\title{Greedy Matching: Guarantees and Limitations}

\author{Bert Besser\thanks{Institut f\"ur Informatik, Goethe-Universit\"at Frankfurt am Main. Email: \texttt{besser@thi.cs.uni-frankfurt.de}. Partially supported by DFG SCHN 503/6-1.}
\and
Matthias Poloczek\thanks{School of Operations Research and Information Engineering, Cornell University. Email:~\texttt{poloczek@cornell.edu}.  Supported by the Alexander von
Humboldt Foundation within the Feodor Lynen program, and in part by NSF grant
CCF-1115256.}
}
\date{}

\maketitle

\input{abstract}

\input{tikzStyles}
\input{intro}
\input{preliminaries}

\section{The \MDeg Algorithm}
\label{section_MinGreedy}

\input{minGreedyHardInstances}
\input{minGreedy_and_variants}

\input{minGreedyOnBoundedDegreeGraphs}

\section{Inapproximability Bounds for Priority Algorithms}
\label{section_inapprox_priority}
%
\input{inapproxTheModel}
\input{inapproxFullyRandomized}
\input{inapproxGreedyBoundedDegree}

\input{inapproxGreedyHyperGraphs}

\input{conclusion}

\bibliographystyle{abbrv}      
\bibliography{matching}

\appendix
\input{appendixMinGreedyLinTime}

\end{document}

%% file: abstract.tex
\begin{abstract}
Since Tinhofer proposed the \MDeg algorithm for maximum cardinality matching
in~1984, several experimental studies found the randomized algorithm to perform
excellently for various classes of random graphs and benchmark instances.
In contrast, only few analytical results are known.
We show that \MDeg cannot improve on the trivial approximation ratio
of~$\frac{1}{2}$ whp., even for bipartite graphs. 
Our hard inputs seem to require a small number of high-degree nodes.

This motivates an investigation of greedy algorithms on graphs with 
maximum degree~\maxdeg:
We show that~\MDeg achieves a~$\desiredRatio$-approximation for graphs
with~$\maxdeg {=} 3$ and for \maxdeg-regular graphs, and a guarantee of~$\weakerRatio$ for graphs with maximum degree~$\maxdeg$.
Interestingly, our bounds even hold for the deterministic \MDeg that breaks
all ties arbitrarily.
%

Moreover, we investigate the limitations of the greedy paradigm, using the model
of \emph{priority algorithms} introduced by Borodin, Nielsen, and Rackoff.
We study deterministic priority algorithms and prove
a~\desiredRatio-inapproximability result for graphs with maximum
degree~$\maxdeg$; thus, these greedy algorithms do not achieve a~$\frac{1}{2}
{+} \varepsilon$-approximation and in particular the~$\frac{2}{3}$-approximation
obtained by the deterministic \MDeg for~$\maxdeg {=} 3$ is optimal in this
class.
For~$k$-uniform hypergraphs we show a tight~$\frac{1}{k}$-inapproximability bound.

We also study fully randomized priority algorithms and give
a~$\frac{5}{6}$-inapproximability bound. Thus, they cannot compete with matching
algorithms of other paradigms.

%


%
%
\end{abstract}

%% file: tikzStyles.tex
\usetikzlibrary{decorations}
\usetikzlibrary{decorations.markings}
\usetikzlibrary{decorations.pathmorphing}
\usetikzlibrary{fixedpointarithmetic}

\tikzset{
 every node/.style={
  circle,
  draw=black,
  minimum size=4mm,
  inner sep=0mm
 },
 opt/.style={
  thick,
  double=white,
  double distance=2pt
 },
 mg/.style={
  thick,
  postaction={
   decorate,
   decoration={
    markings,
    mark=at position 0.5 with {
     \draw (-0.04,0.15) -- (-0.04,-0.15)
           ( 0.04,0.15) -- ( 0.04,-0.15);
    }
   }
  }
 },
 transfer/.style={
  >=triangle 45,
  ->
 },
 nonMcov/.style={
  fill=lightgray
 },
 altpath/.style={
  decoration={snake},
  decorate
 },
 alabel/.style={
  rectangle,
  fill=none,
  draw=none
 },
 comp/.style={
  lightgray,
  line width=6pt
 },
 compopt/.style={
  opt,
  double=lightgray
 },
 mylabel/.style={
  inner sep=0pt,
  fill=none,
  draw=none,
  minimum size=1pt
 }
}

\tikzset{VertexStyle/.style = {shape          = circle,
                                 text           = black,
                                 inner sep      = 2pt,
                                 outer sep      = 0pt,
                                 minimum size   = 12 pt}}
\tikzset{ImplicationStyle/.style = {-latex, thick}}
\tikzset{AlertImplicationStyle/.style = {-latex, very thick, color=#1},
		AlertImplicationStyle/.default = red
												}
\tikzset{EquivalenceStyle/.style = {latex-latex, thick}}
\tikzset{AlertEquivalenceStyle/.style = {latex-latex, very thick, color=#1},
		AlertEquivalenceStyle/.default = red
												}

%% file: intro.tex
\section{Introduction}
\label{section_intro}
Due to their simplicity and efficiency, greedy algorithms have been studied
intensely for the maximum cardinality matching problem, a problem that arises in
many applications
%
including image feature
matching \cite{Cheng96maximum-weightbipartite}, pairwise kidney
exchange~\cite{kidney2,tri12}, protein structure comparison
\cite{bsx08}, and low delay network traffic
routing~\cite{hs07}.

In~1984 Tinhofer~\cite{tin84} proposed the following three randomized greedy
algorithms.
All have in common that iteratively an edge is picked and added to the matching; to
obtain a feasible matching, afterwards both endpoints are removed from the graph
with all their incident edges. Thus, the crucial aspect is how the edge is
selected.
The first algorithm, simply referred to as~\greedy,  picks the edge uniformly at
random among all edges.
The modified randomized greedy algorithm~(\AlgFrieze), however,
selects a node, the so-called first endpoint, uniformly at random and matches
it to a randomly picked neighbor, the second endpoint.
Tinhofer's third algorithm, \MDeg, is identical, except that the first endpoint
is selected uniformly at random among all nodes that have minimum degree.
%
These greedy algorithms can be implemented in linear time using only simple data
structures; for \mingreedy we describe such a data structure in
Sect.~\ref{section_mingreedy_implementation}. A linear time implementation of
\AlgFrieze is proposed in~\cite{ps12}.

Note that~\MDeg can be interpreted as a refinement of~\AlgFrieze, since~\MDeg
prioritizes nodes that have a minimum number of matchable neighbors remaining.
%
And indeed, Tinhofer provides experimental and analytical results for
Erd\H{o}s-R\'enyi random graphs with varying density: here \mingreedy achieves
an expected matching size larger than \AlgFrieze or \greedy.
%
%
Frieze, Radcliffe, and Suen~\cite{frs95} found a superior performance of \MDeg
on random cubic graphs: for instance, \MDeg left only about~$10$ out
of~$10^6$ nodes unmatched, and hence performed better than \AlgFrieze by orders
of magnitude.
%
%
An excellent performance was also observed in experiments for random graphs with
small constant average degree~\cite{Magun97greedymatching}
and on graphs arising from a real world application~\cite{hs07}.

In contrast little is known about rigorous performance guarantees of
\MDeg:
%
%
The most important analytical result is due to Frieze, Radcliffe, and Suen~\cite{frs95} and
states that \MDeg leaves only~$o(n)$ nodes unmatched on random cubic graphs in
expectation.
No worst case analysis is known.
On the other hand, \AlgFrieze beats approximation ratio~$\frac{1}{2}$ by a small
constant~\cite{adfs95,ps12} on general graphs, whereas \greedy
cannot~\cite{df91}.
%
Recently, Chan, Chen, Wu, and Zhao~\cite{ccwz14} showed that the \ranking algorithm
of~\cite{kvv90} achieves at least a~$0.523$-approximation on general graphs; 
we refer to~\cite{my11,kmt11} for results on bipartite graphs and to~\cite{ps12}
for Erd\H{o}s-R\'enyi random graphs.

Our first result is a worst case analysis of \MDeg (see
Sect.~\ref{section_MinGreedy}): We propose a family of bipartite graphs and show
that \MDeg does not achieve approximation ratio~$\frac{1}{2} {+} \varepsilon$
whp.\ for any~$\varepsilon {>} 0$.
Thus, somewhat surprisingly, it performs even worse than the seemingly less
optimized \AlgFrieze algorithm on their respective hardest inputs.
Moreover, we show that two popular randomized variants of \MDeg also fail on
our graphs.

%

The closer the approximation ratio of~\MDeg is pushed towards~$\frac{1}{2}$, the
larger the number of nodes~$n$ must be in order to obtain the respective
upper bound on the approximation ratio.
It turns out that these graphs have a small number of high degree nodes: The
maximum degree is~$\Theta(n)$, whereas the average degree is only~$O(\sqrt{n})$.
We wonder if these nodes of high degree are essential in order to enforce a bad
performance for~\MDeg? This question motivates an investigation of graphs whose
maximum degree~$\maxdeg$ is bounded above by a constant.
For graphs with maximum degree at most three we show that \MDeg achieves an
approximation ratio of~$\frac{2}{3}$ (see Sect.~\ref{subsection_mindegree_bounded_deg_graphs}).
Moreover, we show that~\MDeg gives a~$\desiredRatio$-approximation
on~$\maxdeg$-regular graphs for any~$\maxdeg$.
%
We conjecture this guarantee to hold for all graphs with maximum degree~\maxdeg;
what we show now is a slightly weaker guarantee of \weakerRatio.

%

Note that our guarantees even hold if all ties are broken deterministically,
i.e.\ we select some minimum degree node arbitrarily and match it to an
arbitrary neighbor.
Therefore, it is even more striking that the worst case guarantee for \mingreedy is better
than the best known guarantee for the \emph{expected} approximation ratio of~$0.523$
for~\ranking~\cite{ccwz14} if~$\maxdeg \leq 11$; compared to the best known bound for the expected
approximation ratio of~$0.5000025$~\cite{adfs95} for~\AlgFrieze our
guarantee is better whenever~$\maxdeg \leq 100,001$ holds.
%
%
%
%
Regarding \greedy, our bound is better for all \maxdeg than the known expected
performance guarantee of $\oneOverTwo(\sqrt{(\maxdeg{-}1)^2{+}1}{-}\maxdeg{+}2)$
shown by Miller and Pritikin~\cite{mp97}; in particular, \greedy achieves an
expected~$0.618$-approximation for $\maxdeg=3$.

\medskip

We also study the inherent limitations of greedy algorithms for the maximum
matching problem. 
While the above mentioned algorithms are certainly all ``greedy'' in a very
natural sense, we need a formal characterization of what constitutes a greedy
algorithm for the problem at hand.
To this end, we use the model of \emph{priority algorithms} introduced by
Borodin, Nielsen, and Rackoff~\cite{bnr03}. This model has been applied
successfully to a large range of
problems~\cite{di09,bblm10,poloczek11,biyz12,hb13}, see~\cite{hb13} for a recent
summary.
Greedy algorithms explore the input myopically, making an
irrevocable decision for each piece.
In the vertex model proposed in~\cite{di09,bblm10} each piece of the input, also
referred to as \emph{data item}, corresponds to a node and its neighbors.
An \emph{adaptive priority algorithm} submits an ordering~$\pi$ on all data
items without looking at the input~$\cal I$, and receives the first data item
in~$\cal I$ according to~$\pi$, say for node~$u$.
Then the algorithm either matches~$u$ to one of its neighbors or isolates~$u$.
If the algorithm is additionally required to be \emph{greedy}, then it does not
have the option to isolate~$u$.
%
Note that the deterministic~\MDeg (that breaks all ties arbitrarily) can be
implemented as a greedy adaptive priority algorithm.

We show that no greedy adaptive priority algorithm achieves an approximation
ratio better than~$\desiredRatio$ on graphs with maximum degree~$\maxdeg \geq 3$
(see Sect.~\ref{section_greedy_priority}). In particular, no such deterministic
algorithm can guarantee an approximation ratio of~$\frac{1}{2} + \varepsilon$
for any~$\varepsilon > 0$, therefore providing evidence that randomness is
essential for greedy algorithms in order to guarantee a non-trivial
approximation.
For $k$-uniform hypergraphs we show a tight~$\frac{1}{k}$-inapproximability bound (see
Sect.~\ref{section_hypergraph_matching}).

For priority algorithms that are not required to be greedy we show an
inapproximability bound of~$\frac{2}{3}$ (see
Sect.~\ref{section_nongreedy_adaptive_algorithms_matching}); it relies on a graph of
maximum degree three, hence our performance guarantee for \mingreedy for~$\maxdeg = 3$ is
tight.
We also study randomized priority algorithms introduced in~\cite{ab10} and show
that these do not achieve an expected approximation ratio better
than~$\frac{5}{6}$ (see Sect.~\ref{section_randomized_priority}).
Therefore, these randomized greedy algorithms cannot achieve the performance of
algorithms based on augmenting paths or algebraic methods.
We point out that our class of randomized priority algorithms is very comprehensive; in particular, it contains all algorithms mentioned so far as well as the randomized algorithm of Karp and Sipser~\cite{ks81,afp98}.
Furthermore, our class subsumes the models of randomized algorithms proposed by Goel and Tripathi~\cite{gt12}.
In particular, their models do not contain \MDeg.

%% file: preliminaries.tex
\subsection{The Maximum Cardinality Matching Problem}
\label{section_prelim}
Let~$ \g=(V,E) $ be an unweighted and undirected graph.
A matching~$ \m\subseteq E $ is a selection of edges of~\g such that no two edges in~\m share a node.
If no further edge of~$E \setminus M$ can be added to~\m without violating this condition, then~\m is called a \emph{maximal} matching.
If~\m has largest cardinality among all matchings in~$G$, then~\m is a~\emph{maximum} matching.
If~\m covers all nodes, we call it \emph{perfect}.

An algorithm is called an~$\mgapprox$-approximation algorithm for the maximum (cardinality) matching problem, if for any graph it finds in polynomial time a matching whose size is at least~$\mgapprox$ times the size of a maximum matching.
A randomized~$\mgapprox$-approximation algorithm is a polynomial time algorithm that obtains a matching whose \emph{expected} cardinality is at least~$\mgapprox$ times the optimal solution.

Note that any maximal matching~$\m$ has size at least half of the optimum. To see this, choose any maximum matching~$\mopt$ and observe that each edge in~$\m$ shares nodes with at most two edges in~$\mopt$.
Thus, any algorithm that finds a maximal matching is a~$\frac{1}{2}$-approximation algorithm. In particular, the greedy algorithms we consider in this article will have this property.
%
%
%
%
%

\paragraph*{Related Work outside the Greedy Paradigm.}
~A large variety of algorithmic techniques have been applied to the maximum matching problem. 
We provide a brief outline.

Edmonds~\cite{edm65} proposed the famous blossom algorithm that finds a maximum
matching in polynomial time by detecting augmenting paths.
Gabow~\cite{gabow76} showed how to implement this algorithm with running time~$O(|V|^3)$.
Micali and Vazirani gave an exact algorithm with running time~$O(\sqrt{|V|}
|E|)$, that iteratively finds short augmenting paths. However, its analysis and implementation are  non-trivial (cp.~\cite{vaz13,gab14} and further references therein).
%

Goldberg and Karzanov~\cite{gk04} proposed an algorithm based on flow techniques
with a running time of~$O(\sqrt{|V|} |E| \log{(n^2 / m)}/\log{n})$; to the best of our knowledge this is the fastest exact algorithm for dense non-bipartite graphs.

Using algebraic methods, a maximum matching can be found in time~$O(|V|^\omega)$
by the randomized algorithm of Mucha and Sankowski~\cite{ms04} (see
also~\cite{geelen00,harvey09}), where~$\omega$ is the exponent of the best known
matrix multiplication algorithm.
%
%

%% file: minGreedyHardInstances.tex
\label{subsection_hard_graphs_mindegree}
Encouraged by its excellent performance on random graphs, we study the worst
case performance of \MDeg and present a family of hard graphs.
In particular, our graphs show that the expected approximation ratio of \mingreedy is at most~$ \oneOverTwo+\varepsilon $ whp. for any~$ \varepsilon>0 $.

Choose $a,b \in \mathbb{N}$, where~$b$ is even and~$b \mid a$, and let 
$c := 2\cdot\left\lceil\sqrt{a}\right\rceil$.
The graph~$G_{a,b}$ has $2a {+} c$ nodes partitioned as follows:
\begin{align*}
S_1 &= \{1, 2, \dots, a\}\\
S_2 &= \{a+1, a+2, \ldots, 2a\},\text{ and}\\
S_3 &= \{2a+1, 2a+2, \ldots, 2a+c\}.
\end{align*}
Edges are as follows (see Fig.~\ref{fig_mingreedy} for an example of the graph~$ G_{a,4} $):
\begin{enumerate}[label=\roman{*}.]
  \item\label{s1s3} Each node in $S_1$ is connected to \emph{every} node in
  $S_3$.
  \item\label{s1s2} Each node $i \in S_1$ is connected to node $a+i \in S_2$.
  \item\label{s3s3} Each $ S_3 $-node is connected to exactly one $ S_3 $-node.
  \item\label{s2s2} The crucial ingredient is: within $S_2$ we form disjoint
  cliques of size $b$.
\end{enumerate}

\begin{figure}[h]
\centering
\scalebox{.75}{
\input{figures/ga4.tex}
}
\caption[$G_{a,4}$]{The graph $G_{a,4}$. The edges represented by double lines form a perfect matching
}
\label{fig_mingreedy}
\end{figure}
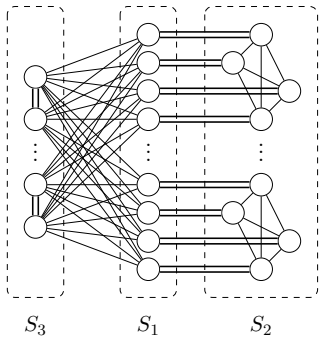

Observe that a perfect matching of size $ a+\lceil\sqrt{a}\rceil $ consists of
the $ a $ edges connecting the $ S_1 $-node with the $ S_2 $-nodes and the $
\lceil\sqrt{a}\rceil $ edges within $ S_3 $.
The edges of the perfect matching are displayed as double edges in Fig.~\ref{fig_mingreedy}.

Our intention is that an instantiation of the \MDeg algorithm matches $ S_2
$-nodes with $ S_2 $-nodes initially, since these nodes will have minimum
degree.  Thereby the respective algorithm may add $ \frac{|S_2|}{2}=\frac{a}{2}
$ edges within the cliques of $ S_2 $ to its matching.
When all $S_2$ nodes are matched, all remaining edges are incident with nodes
in~$S_3$.
%
Hence at most $ |S_3|=2\lceil\sqrt{a}\rceil $ additional edges can be added to
the matching. If the randomized algorithm follows this intended scheme, then it
achieves an approximation ratio of $
(\frac{a}{2}+2\lceil\sqrt{a}\rceil)/(a+\lceil\sqrt{a}\rceil) $ on the
graph~$G_{a,b}$, which tends to $ \oneOverTwo $ for large $ a $.
However, it turns out that the randomized \MDeg algorithm does not strictly
adhere to this plan.

%% file: figures/ga4.tex
\begin{tikzpicture}
\node (s3_1) {};
\node (s3_2) at ($ (s3_1)+(0,-.75) $)[] {};
\node[alabel, below=0mm of s3_2] (dots) {$ \vdots $};
\node (s3_3)[below=2mm of dots] {};
\node (s3_4)at ($ (s3_3)+(0,-.75)$) {};

\node at($ (s3_1)+(2,.75) $) (s1_1) {};
\node[] at ($(s1_1)+(0,-.5)$) (s1_2) {};
\node[] at ($(s1_2)+(0,-.5)$) (s1_3) {};
\node[] at ($(s1_3)+(0,-.5)$) (s1_4) {};
\node[alabel, below=0mm of s1_4](dots) {$ \vdots $};
\node[, below=2mm of dots](s1_5) {};
\node[]at ($(s1_5)+(0,-.5)$) (s1_6) {};
\node[]at ($(s1_6)+(0,-.5)$) (s1_7) {};
\node[]at ($(s1_7)+(0,-.5)$) (s1_8) {};

\node at($ (s1_1)+(2,0) $) (s2_1) {};
\node[] at($ (s2_1)+(-.5,-.5) $)(s2_2) {};
\node[] at($ (s2_1)+(+.5,-1) $)(s2_3) {};
\node[] at($ (s2_1)+(0,-1.5) $)(s2_4) {};
\node[alabel, below=0mm of s2_4] (dots) {$ \vdots $};
\node[, below=2mm of dots](s2_5) {};
\node[]at($ (s2_5)+(-.5,-.5) $) (s2_6) {};
\node[]at($ (s2_5)+(.5,-1) $) (s2_7) {};
\node[]at($ (s2_5)+(0,-1.5) $) (s2_8) {};

\draw[opt] (s3_1) -- (s3_2);
\draw[opt] (s3_3) -- (s3_4);

\draw[opt] (s2_1) -- (s1_1);
\draw[opt] (s2_2) -- (s1_2);
\draw[opt] (s2_3) -- (s1_3);
\draw[opt] (s2_4) -- (s1_4);

\draw[opt] (s2_5) -- (s1_5);
\draw[opt] (s2_6) -- (s1_6);
\draw[opt] (s2_7) -- (s1_7);
\draw[opt] (s2_8) -- (s1_8);

\draw[] (s2_1) -- (s2_2);
\draw[] (s2_1) -- (s2_3);
\draw[] (s2_1) -- (s2_4);
\draw[] (s2_2) -- (s2_3);
\draw[] (s2_2) -- (s2_4);
\draw[] (s2_3) -- (s2_4);

\draw[] (s2_5) -- (s2_6);
\draw[] (s2_5) -- (s2_7);
\draw[] (s2_5) -- (s2_8);
\draw[] (s2_6) -- (s2_7);
\draw[] (s2_6) -- (s2_8);
\draw[] (s2_7) -- (s2_8);

\draw (s3_1) -- (s1_1);
\draw (s3_1) -- (s1_2);
\draw (s3_1) -- (s1_3);
\draw (s3_1) -- (s1_4);
\draw (s3_1) -- (s1_5);
\draw (s3_1) -- (s1_6);
\draw (s3_1) -- (s1_7);
\draw (s3_1) -- (s1_8);
\draw (s3_2) -- (s1_1);
\draw (s3_2) -- (s1_2);
\draw (s3_2) -- (s1_3);
\draw (s3_2) -- (s1_4);
\draw (s3_2) -- (s1_5);
\draw (s3_2) -- (s1_6);
\draw (s3_2) -- (s1_7);
\draw (s3_2) -- (s1_8);
\draw (s3_3) -- (s1_1);
\draw (s3_3) -- (s1_2);
\draw (s3_3) -- (s1_3);
\draw (s3_3) -- (s1_4);
\draw (s3_3) -- (s1_5);
\draw (s3_3) -- (s1_6);
\draw (s3_3) -- (s1_7);
\draw (s3_3) -- (s1_8);
\draw (s3_4) -- (s1_1);
\draw (s3_4) -- (s1_2);
\draw (s3_4) -- (s1_3);
\draw (s3_4) -- (s1_4);
\draw (s3_4) -- (s1_5);
\draw (s3_4) -- (s1_6);
\draw (s3_4) -- (s1_7);
\draw (s3_4) -- (s1_8);

\draw[dashed,rounded corners] ($(s1_1)+(-.5,+.5)$) rectangle ($(s1_8)+(+.5,-.5)$);
\draw[dashed,rounded corners] ($(s2_1)+(-1,+.5)$) rectangle ($(s2_8)+(+1,-.5)$);
\draw[dashed,rounded corners] ($(s1_1)+(-2.5,+.5)$) rectangle ($(s1_8)+(-1.5,-.5)$);

\node[alabel] (s1) at ($ (s1_8)+(0,-1) $) {$ S_1 $};
\node[alabel]  (s2)at ($ (s1)+(2,0) $) {$ S_2 $};
\node[alabel] (s3)at ($ (s1)+(-2,0) $) {$ S_3 $};
\end{tikzpicture}

%% file: minGreedy_and_variants.tex
\subsection{A Worst Case Analysis of \MDeg and two Variants}
\label{subsection_mindegree_uniform_TB}
Hougardy~\cite{h09} considered deterministic algorithms that iteratively pick an edge~$(u,v)$ such that either~$u$ or~$v$ has currently minimum degree. He proposed a family of graphs for which any such algorithm achieves an approximation ratio of~$\frac{1}{2} + o(1)$. 
Since the difficulty of these graphs relies on the assumption that all ties are broken towards the worst case, they are no obstacle for~\mingreedy.
%

We begin with a worst case study of Tinhofer's \MDeg. 
In each iteration the algorithm selects a node, the first endpoint, uniformly at
random among all nodes of minimum degree and matches it to the second endpoint; this neighbor is also picked uniformly at random.
It is a well-known fact that for every node~$v$ in a connected
graph there is a maximum matching that covers~$v$ 
(see~\cite{ef65} for a stronger statement).
This provides additional motivation for \MDeg, since we have a good chance to
pick an optimal edge if the first endpoint has small degree. 
%
%
However, we show:
%
\begin{theorem}
\label{theorem_mindegree_randomized}
\MDeg cannot approximate the Maximum Matching Problem within~$\frac{1}{2} +
\varepsilon$ whp.\ (for any~$\varepsilon > 0$), even on bipartite graphs.
%
\end{theorem}
\begin{proof}
We consider the graph~$G_{a, \sqrt{a}}$, and let~$a$ be a large (integral)
square number for the sake of convenience.
The crucial property~$P$ is that all nodes of minimal degree are
in~$S_2$. Initially, this is truly the case: 
Recall that by construction every node in~$S_2$ has degree~$\sqrt{a}$, whereas
the degrees are~$2\sqrt{a}+1$ in~$S_1$ and $a+1$ in~$S_3$.

In the following, the notion of a \emph{$(S_j,S_k)$ matching}, for~$k,j \in
\{1,2,3\}$,  means that a node in set~$S_j$ is selected by the algorithm because
of its minimal degree and matched to a neighbor in set~$S_k$.
When may the property~$P$ be violated?
\begin{itemize}
\item Every $(S_2,S_2)$ matching reduces the degree of their clique-neighbors by
2 and the degrees of their \emph{unique} $S_1$ neighbors by 1. Since any node in
$S_1$ is affected only once, this case does not pose any danger.
\item Every $(S_2,S_1)$ matching reduces the degree of all nodes in $S_3$ and
all nodes in the respective clique by 1, but it has no effect on any other node in
$S_1$. Hence, a large number of $(S_2,S_1)$ matchings reduces the degree of nodes in
$S_3$ from its initial value $a+1$ below the degrees of nodes in $S_2$ (which is
at most $\sqrt{a}$). In that event the algorithm prefers matching the nodes in
$S_3$, thereby drastically reducing the degrees in $S_1$. As a consequence,
a large number of nodes in $S_1$ might appear in the matching and the algorithm
may come up with a good approximation.
We will see, however, that this scenario is very unlikely.
\end{itemize} 
First we assume that~$P$ is preserved. Later we show that~$P$ holds whp.\
throughout the computation.
We need to bound the expected number of~$(S_2,S_1)$ matchings.
Observe that the algorithm processes the cliques of~$S_2$ one after another,
fixing all nodes of a clique before moving on to the next.

At first we focus on a single clique in~$S_2$ that initially has~$\sqrt{a}$
nodes.
The probability that a node $v \in S_2$ is matched to its neighbor $u \in S_1$,
conditioned on the event that \MDeg selects~$v$ as first endpoint,
%
depends on the number~$x$ of $(S_2,S_2)$ matchings and the number~$y$ of
$(S_2,S_1)$ matchings that have been performed in the respective clique so
far.
In particular, the conditional probability equals $\frac{1}{\text{deg}(v) - 2x -
y} = \frac{1}{\sqrt{a} - 2x - y}$, since~$v$ has~$\sqrt{a}-1$ neighbors in~$S_2$
and one in~$S_1$.
%
%
%
Then the expected number of $(S_2, S_1)$ matchings (for a single clique) is
given by the sum of such probabilities, where \mbox{$0 \leq 2x+y < \sqrt{a}$} holds.

For what values of~$x$ and~$y$ a random trial is performed depends on the
outcome of previous trials.
We attain a worst case perspective. Let~$Z_i$ for~$0 \leq i < \sqrt{a}$ be a
binary random variable with success probability~$p_i = \frac{1}{\sqrt{a} - i}$.
Then~$\sum_i{Z_i}$ is an upper bound on the number of $(S_2, S_1)$ matchings for
a single clique.
If~$a$ is sufficiently large, the expected value is upper-bounded by
$\E\left[\sum_i{Z_i} \right] = \sum_{i=1}^{\sqrt{a}}{\frac{1}{i}} \leq \ln(a)$,
where it is implicitly assumed that all previous actions were $(S_2,S_1)$
matchings.

%
To show that property~$P$ is likely to be preserved, we show that~$\sum_i{Z_i}$
is concentrated at its expected value, using a variant of the Chernoff bound
(cp.\ Sect.~1.6 in~\cite{dp09}).
Then,
\begin{equation*}
\mathrm{prob}\left[\sum_i{Z_i} \geq
\left(1+\beta\right)\cdot \ln(a)\right] 
\leq \left(\frac{e^{\beta}}{\left(1+\beta\right)^{(1+\beta)}}\right)^{\ln(a)}
\leq \frac{1}{a^\frac{3}{2}}
\end{equation*}
where the second inequality follows by choosing~$\beta = e^2$.
We employ the Union bound to obtain a lower bound of
$1-\sqrt{a} / a^{\frac{3}{2}} = 1 - \frac{1}{a}$
on the probability that for none of the~$\sqrt{a}$ cliques inside~$S_2$ the sum
deviates by more than a factor of~$1 + \beta$ from its expected value.

Hence, there are whp.\ at most~$O\left(\sqrt{a} \cdot \ln(a)\right) = o(a)$
$(S_2,S_1)$ matchings in total and in particular property~$P$ is preserved whp.\
throughout the course of the algorithm, since a violation would require
$\Omega(a)$ $(S_2,S_1)$ matchings.
Furthermore, the size of the matching returned by \MDeg on~$G_{a, \sqrt{a}}$ is
at most~$\frac{a}{2} + o(a)$ with high probability, whereas the maximum matching
in the graph has cardinality at least~$a$.

To obtain a bipartite graph, we take two copies~$L,R$ of~$G_{a, \sqrt{a}}$ and
let~$S_i := S_i^L {\cup} S_i^R$. First we remove the edges within~$S_3$ and
create an edge for each node in~$S_3^L$ and its counterpart in~$S_3^R$.
Then we remove all edges within~$S_2$ and form complete bipartite cliques of
each former clique in~$S_2^L$ and its counterpart in~$S_2^R$. Note that the
degree of nodes in~$S_2$ have increased by one, which is negligible in our
analysis. All other degrees do not change.
The maximum matching increased by at least~$a$ edges, whereas the output 
matching grows by~$\frac{a}{2} {+} o(a)$.
\end{proof}
There are two natural variants of \mingreedy.
The first variant, \EDSM (enhanced Degree-Sequenced Matching), matches a node of
minimum degree to a neighbor that has in turn minimum degree among all neighbors.
All ties are broken uniformly at random.
Hosaagrahara and Sethu \cite{hs07} apply \EDSM to the problem of assigning
network packets from input ports to output ports to achieve low average delay
and find that it performs very well on real and synthetic traffic traces.
Another variant proposed in \cite{hs07} reduces to \EDSM on our hard
instances.
\MDS (Minimum Degree Sum) picks the next edge such that the sum of the degrees
of its endpoints is minimal, also breaking ties at random.
Both algorithms ignore edges incident with already matched nodes.

\EDSM and \MDS achieve only a trivial expected approximation ratio, as we show next.
For \MDS, this fact also follows from the definition of graphs given in \cite{h09}.
We present a unified construction that applies to both algorithms.

\begin{theorem}
\label{theorem_mindegree_varianten}
For any~$\varepsilon > 0$, there is a bipartite graph such that neither \EDSM nor
\MDS achieve approximation ratio~$\frac{1}{2} + \varepsilon$.
\end{theorem}
\begin{proof}
As worst case input we use the graph~$G_{a,2}$. See the beginning of
Sect.~\ref{subsection_hard_graphs_mindegree} for description of the family of graphs,
and our intention of how the algorithm should proceed.
%
%
We show below how to make the graph bipartite.
Recall that the degrees of $ S_3 $- and $ S_1 $-nodes increase with $ a $,
whereas all $ S_2 $-nodes have degree two.

As desired, both \EDSM and \MDS start by repeatedly matching a node in $ S_2 $
with its unique neighbor in $ S_2 $ until all nodes in $ S_2 $ are matched.
Thus, the claimed bound on the approximation ratio follows from the discussion
at the beginning of Sect.~\ref{subsection_hard_graphs_mindegree}.

To obtain a bipartite graph, we first remove all inner edges of~$S_3$, thereby
decreasing the size of the maximum matching
by~$\left\lceil\sqrt{a}\right\rceil$.
Note that if all nodes of~$S_1$ were connected to all nodes in~$S_3$, there
would be cycles of odd length. Hence we add a set~$S_3'$ of nodes,
with~$|S_3'|=|S_3|=2\left\lceil\sqrt{a}\right\rceil$, and connect odd $ S_1
$-nodes with $ S_3 $-nodes and even $ S_1 $-nodes with $ S_3' $-nodes such that
these edges are evenly distributed over the nodes of $ S_3 $ and $ S_3' $.
The matching obtained by the algorithm is increased by~$o(a)$ only. 
\end{proof}
%
%

%% file: minGreedyOnBoundedDegreeGraphs.tex
\subsection{Guarantees for \MDeg on Bounded Degree Graphs}
\label{subsection_mindegree_bounded_deg_graphs}
In the construction of Theorem~\ref{theorem_mindegree_randomized} the
approximation ratio of \mingreedy converges to~\oneOverTwo as the maximum degree
increases. In this section we study degree bounded graphs.

We state that if all degrees are bounded above by~$ \maxdeg=2 $, then any
algorithm that picks an edge incident to a node of degree one (if such a node
exists) computes a maximum matching.
%
Thus, in particular \mingreedy and \karpsipser~\cite{ks81}, that picks a random
edge incident with a degree-1 node, if such an edge exists, and a random edge
otherwise, are optimal on such graphs.

In what follows we show that if the degrees in the input are bounded above
by~$\maxdeg \geq 3$, then the approximation ratio of \mingreedy is strictly
better than~$\frac{1}{2}$.
Interestingly, our guarantee holds for any algorithm that iteratively picks an
edge that is incident with some node of current minimum degree.
%
Thus, we do not use the feature that \mingreedy breaks all ties uniformly at random.
In particular, our bounds hold for the deterministic variant of
\mingreedy that picks an arbitrary non-isolated node of minimum degree and matches it with
an arbitrary neighbor. 
%
%
%
Then both nodes and all their incident edges are removed from the graph. Therefore each node left in the graph is unmatched, and the algorithm iterates on the remaining vertices.

First we study graphs of bounded degree~$ 3 $ and show that the deterministic variant of \mingreedy, achieves an approximation ratio of~$\frac{2}{3}$.
Along the same lines, we also obtain a guarantee of~$\desiredRatio$ for
\maxdeg-regular graphs.

In Sect.~\ref{section_mingreedy_guarantee_larger_maxdegs} we prove a slightly worse bound for graphs of maximum degree~\maxdeg.

\begin{theorem}
\label{thm:threeDeb}
If all degrees are bounded above by $ \maxdeg=3 $ or the input graph is \maxdeg-regular with $ \maxdeg\geq4 $, then \mingreedy achieves approximation ratio at least $\desiredRatio$.

More generally, these bounds hold for any algorithm which selects edges incident with a node of minimum degree.
\end{theorem}

Let $\g=(V,E)$ be a connected graph and \m the matching computed by \mingreedy.
Given~\m we choose a maximum matching \mopt for the analysis such that the connected components of the graph $ H=(\m\cup\mopt) $ are of two types:
%
%
An \emph{\mmoptpath} $ X $ has $ \w{X}\geq1 $ edges of \m and $ \wopt{X}=\w{X}+1 $ edges of \mopt that alternate. In particular, any augmenting path starts and ends with an \mopt-edge.
%
%
%
The second type are \emph{\oneonepaths}. Such a component is a path of length one in~$H$, i.e.\ its edge belongs both to~\m and~\mopt.
%
%
%

Why can we always find such an optimal matching~\mopt? Observe that~$H$ has maximum
degree two, hence it consists of cycles and paths. Every component that is not a singleton alternates
between~\m and~\mopt, since every node has degree at most one in each matching. 
Then the crucial observation is that each component~$ X $ that is either an even-length cycle
or an even-length path can be replaced by~$ \w{X} $ \oneonepaths; we simply exchange the respective edges of~\mopt by the ones of~\m.
%
%
%

\inlineheading{Local Approximation Ratios.}
To obtain a bound of $ \mgapprox=|\m|/|\mopt|\geq\desiredRatio $ on the global approximation ratio $ \mgapprox $ of \mingreedy we want to bound \emph{local approximation ratios}
\begin{align*}
\mgapprox_X=\w{X}/\wopt{X}
\end{align*}
of components $ X $ of \mg.
We call an \mmoptpath of length three a \emph{\onetwopath}.
A \onetwopath~$ X $ has local approximation ratio only $ \mgapprox_X=\oneOverTwo<\desiredRatio $, whereas longer augmenting paths $ X $ have $\mgapprox_X\geq\twoOverThree\geq\desiredRatio$ since~$\maxdeg \geq 3$. A singleton~$ X $ even has~$\mgapprox_X=1$.
Therefore we balance local approximation ratios.
We say that a component $ X $ \emph{has \m-funds} \w{X} and introduce a change $ t_X $ to the \m-funds of $ X $ such that 
\begin{align*} \mgapprox_X=\frac{\w{X}+t_X}{\wopt{X}~~~~~~~}\stackrel{!}{\geq}\desiredRatio\end{align*}
holds for the new local approximation ratio of $ X $.
By transferring \m-funds between components we assert that $ \sum_Xt_X=0 $ holds. Then the total \m-funds $ \sum_X \w{X}+t_X=\sum_X\w{X}=|\m|$ are unchanged and \mingreedy achieves approximation ratio at least
$$ \mgapprox=|\m|/|\mopt|=\left(\sum_X\w{X}+t_X\right)/|\mopt|\geq\left(\sum_X\desiredRatio\cdot\wopt{X}\right)/|\mopt|=\desiredRatio\,. $$

\inlineheading{The Choice of Transfers.}
Our approach is to transfer~\mbox{\m-funds} between components using a selection of edges in~$ F:=E\setminus(\m\cup\mopt)$.
We develop a charging scheme that transfers~\mbox{\m-funds} from~\m-covered nodes to adjacent endpoints of \mmoptpaths. 
Note that~$F$ does not contain edges between endpoints of \mmoptpaths, since otherwise~\m would not be maximal.

In particular, we will assert that components whose local approximation ratio is insufficient receive the required~\mbox{\m-funds}. 
These components in need are \onetwopaths, i.e.\ augmenting paths of length three.

First we show that every endpoint of an \mmoptpath is incident to at least one edge in~$F$.
Let~$d_\g(w)$ denote the degree of node~$w$ in~$\g$.
\begin{lemma}
\label{obs:potentialTransfer}
Let~$w$ be an endpoint of an \mmoptpath. Then $ d_\g(w)\geq2 $ holds.
\end{lemma}
\begin{proof}
Let~$ X $ be the \mmoptpath of~$ w $.
Consider the step when \mingreedy picks the first~\m-edge~$ e $ in~$ X $.
Edge~$ e $ and both its adjacent~\mopt-edges must still be contained in the graph.
%
Since \mingreedy selects a minimum degree node, all nodes of~$X$ have degree at least two.
\end{proof}
Which edges in~$F$ are selected to transfer funds?
%
%
%
\begin{definition}
\label{def:transfer}
Let~$ \edge{v,w}$ be an edge in~$F$, where~$v$ is covered by~\m and~$w$ is an endpoint of an \mmoptpath. Assume that~$ v $ is matched in step~$ s $.
Then edge~\edge{v,w} is a \emph{transfer} if at the end of step~$ s $ the degree of~$ w $ is~$ d(w)\leq\maxdeg-2 $.
%
\end{definition}
We frequently denote a transfer as~$ vw $ to stress its direction from the~\m-covered node~$ v $ to the augmenting path endpoint~$ w $.
In order to refer to transfers to and from a given component $ X $, we also call~$ vw $ a \emph{credit} to $ w $ respectively a \emph{debit} to $ v $.
%
%
%
%
%
%
%
%
\textbf{}
Next we show that each endpoint of an \mmoptpath is guaranteed to receive transfers.

\begin{lemma}
\label{lemma_credits}
Let degrees be bounded by \maxdeg.
The number of credits $ c_w $ to an \mmoptpath endpoint $ w $ is at least $ c_w\geq\min\{d_\g(w)-1,\maxdeg-2\}\geq1$.
\end{lemma}
\begin{proof}
By Lemma~\ref{obs:potentialTransfer} we have $ d_\g(w)\geq2 $, thus $ \min\{d_\g(w)-1,\maxdeg-2\} \geq1$ holds.

If $ d_\g(w)\leq\maxdeg-1 $ holds, then any time the degree of $ w $ drops, it drops to at most~$ \maxdeg-2 $.
Therefore all $ F $-edges incident with $ w $ are credits.

If $ d_\g(w)=\maxdeg $ holds, then after $ d(w) $ drops to $ d(w)=\maxdeg-1 $, all~$ F $-edges of~$ w $ removed later are credits:
at most one $ F $-edge of $ w $ is not a credit.
\end{proof}

\inlineheading{Bounding Local Approximation Ratios.}
Let $ d_X $ (resp., $ c_X $) denote the numbers of debits (respectively credits) to the nodes of a component $ X $.
We call $ d_X-c_X $ the \emph{balance} of~$ X $.
Our approach is to move a constant amount~\transferred of~\m-funds along each transfer.
Assuming that $ d_X-c_X\leq T_X $ holds, the local approximation ratio of $ X $ is at least
\begin{align*}
\mgapprox_X=
\frac{\w{X}-\transferred d_X+\transferred c_X}{\wopt{X}}
\geq\frac{\w{X}-\theta T_X}{\wopt{X}}\,.
\end{align*}
To prove Theorem~\ref{thm:threeDeb} we find \transferred and $ T_X $ such that $ \mgapprox_X\geq\desiredRatio $ holds for all $ X $.

The maximum \emph{possible} number of debits to a component depends on its number of \m-covered
nodes and on \maxdeg:
%
%
By the degree constraint, each \m-covered node of~$X$ has at most~$\maxdeg$ debits.
However, no funds are moved over edges in~$\m {\cup} \mopt$, and if~$X$ is an augmenting path, then
its path endpoints do not incur any debits at all.
We say that a component has \emph{two missing debits}, if its actual number of debits is at least two less than the maximum possible number.
%
%

First consider a \oneonepath $ X $.
By definition, $ X $ does not get credits, i.e.\ $ c_X=0 $.
By the above considerations, both nodes of $ X $ have at most $ \maxdeg-1 $ debits
each.
Using two missing debits, the balance of $ X $ and the local approximation ratio of $ X $ are
\begin{align}
d_X-c_X&=d_X\leq2(\maxdeg-1)-2\label{totalDebits11}\\
\mgapprox_X&=\frac{1-\transferred(d_X-c_X)}{1}\geq 1-2\transferred(\maxdeg-2)\,.\tag{1'}\label{locApprox11}
\end{align}
Note that $ \mgapprox_X\geq\desiredRatio $ holds if we choose $ \transferred\leq\frac{1}{2(2\maxdeg-3)} $.

Let $ X $ be an \mmoptpath.
Since degrees are at most $ \maxdeg=3 $ or the graph is~$ \maxdeg $-regular, by
Lemma~\ref{lemma_credits} credits to a path
endpoint $ w $ of $ X $ (respectively in total to both path endpoints of~$ X $) are
\begin{align}
\label{minCredsPath} c_w\geq\maxdeg-2\mbox{~~~~~respectively~~~~~} c_X\geq2\cdot(\maxdeg-2)\,.
\end{align}
From the above considerations $ d_X\leq2m_X(\maxdeg-2) $ follows.
Assuming that two debits are missing, we obtain
\begin{align}
\label{totalDebitsmmopt}
d_X-c_X&\leq2\w{X}(\maxdeg-2)-2(\maxdeg-2)-2\\
\label{locApproxmmopt}\mgapprox_X&=\frac{\w{X}-\transferred(d_X-c_X)}{\wopt{X}}
\geq\frac{ \w{X}-2\transferred\w{X}(\maxdeg-2)+2\transferred(\maxdeg-1)}{\w{X}+1}\tag{3'}\\
&=1-2\transferred(\maxdeg-2)+\frac{2\transferred(2\maxdeg-3)-1}{\w{X}+1}
=1-2\transferred(\maxdeg-2)\,,\notag
\end{align}
where we choose $ \transferred=\frac{1}{2(2\maxdeg-3)} $ in the last equality.

Combining Eq.~(\ref{locApprox11})~and Eq.~(\ref{locApproxmmopt}) yields our claimed performance guarantee for \mingreedy:
the local approximation ratio of each component $ X $ is $\mgapprox_X\geq 1-2\transferred(\maxdeg-2)=1-\frac{2(\maxdeg-2)}{2(2\maxdeg-3)}=\desiredRatio\,.$

\subsection*{The Proof of Theorem~\ref{thm:threeDeb} for~\maxdeg-Regular Graphs with~$\maxdeg \geq 4$}
%
%
First we show that each component~$X$ has two missing debits, be it due to a node in~$X$ having low degree in \g or being incident to non-transfer $ F $-edges.
The following steps of \mingreedy are crucial:
We say that a node~$u$ \emph{creates} its component~$X$, if \mingreedy selects node~$u$ as the first
endpoint of the first edge matched in~$X$.

\begin{lemma}
\label{lemma:lowDebs}
Let~$ u  $ create~$X$ and~$ d(u) $ be its degree at creation.
$X$ has two missing debits if 
\begin{enumerate*}[label=\alph*),topsep=0mm,noitemsep]
\item\label{lowDebsA}
$ d(u){\leq}\maxdeg-2 $,~~~~
\item\label{lowDebsB}
$ d(u){=}\maxdeg $,~~~~
\item\label{lowDebsC}
$ \maxdeg{\geq}4 $,~~~or~~~~
\item\label{lowDebsD}
$ X $ is a \oneonepath.
\end{enumerate*}
\end{lemma}
\begin{proof}
In order to obtain a contradiction, we assume that at most one debit to a node
of $ X $ is missing.
Let $ X $ be created in step $ s $ when $ u $ is selected because of its minimum degree~$d(u)$ and matched to $ v $.
By assumption, one of $ u,v $ has a debit.

\noindent
If $ d(u){\leq}\maxdeg{-}2 $ at step $ s $, then $ u $ has two missing debits by Definition~\ref{def:transfer}.
This proves part \ref{lowDebsA}, and we may assume $d(u)\geq\maxdeg-1 $ from now on.
Here is an overview of the argument.
Consider a debit $ xw' $ to $ x\in\{u,v\} $.
By Definition~\ref{def:transfer}, step $ s $ removes edge $ \edge{x,w'} $ from \g and, by definition of transfers, after step $ s $ the degree of $ w' $ is $ d'(w')\leq\maxdeg-2 $.
In particular, after step $ s $ the degree of $ w' $ is smaller than the degree of $ u $ before step $ s $.
Our crucial claim is:
\begin{quote}
\emph{At step $ s{+}1 $ there is an \mmoptpath endpoint $ w $ that currently has
minimum degree~$d'(w)$ with~$\maxdeg - 2 \geq d'(w) \geq 1 $. In particular, $ w $ is not isolated.} 
\end{quote}
Note that node~$w$ has a neighbor in~$X$, but does not necessarily belong to~$X$.
Since~$ w $ is an \mmoptpath endpoint, it cannot be the node selected by \mingreedy in step $ s+1 $.
Hence it cannot be the case that at step $ s+1 $ all minimum degree nodes are \mmoptpath endpoints:
There must be a node $ y\neq w $ that is selected in step~$ s+1 $ with minimum degree $ d'(y)=d'(w) \leq\maxdeg-2$.
Since the degree of $ y $ was at least~$ d(u) {=} \maxdeg-1 $ before step $ s $, in step $ s $ edges of $ F $ connecting $ y $ with $ u $ or $ v $ are removed.
But $ u,v,y $ are \m-covered, hence the removed edges are not debits to $ u,v $.
We obtain a contradiction if we find at least two missing debits.

\ref{lowDebsB}
Assume that $ d(u)=\maxdeg $ holds at step $ s $.
Since in step $ s $ at most two edges incident with $ w $ are removed and there is a transfer~$ xw $ with~$ x\in\{u,v\} $, we have $ d'(w)=\maxdeg-2 $ at step $ s+1 $.
In particular, our claim holds:
node $ w $ is not isolated, since $ d'(w)=\maxdeg-2\geq1 $, and $ w $ has minimum degree.
The reason is that every node had degree~$\maxdeg$ before step~$s$, and every
degree drops at most by two.
So let $ y\neq w $ be the node selected in step $ s+1 $.
The degree of $ y $ also drops from \maxdeg to $ \maxdeg-2 $ in step $ s $ when incident edges $ \edge{u,y},\edge{v,y} $ are removed.
If $ y $ is not a node of $ X $, then both $ \edge{u,y},\edge{v,y} $ are $ F $-edges but they are not debits, since $ u,v,y $ are \m-covered. 
Hence we have found for each of $ u,v $ a missing debit.
If $ y $ is a node of $ X $, then $ u,v,y $ form a triangle and one of $ \edge{u,y},\edge{v,y} $, say $\edge{u,y}$, is an $ F $-edge. 
Since both $ u,y $ are \m-covered, edge $\edge{u,y}$ is not a transfer and
hence both $ u,y $ have a missing debit.

\ref{lowDebsC}
Since parts \ref{lowDebsA} and \ref{lowDebsB} apply to $ \maxdeg\geq4 $, we may assume that $ d(u)=\maxdeg-1 $ holds at step $ s $.
%
Recall that node $ u $ has a missing debit due to its low degree.
We study step $ s+1 $ to find an additional missing debit.
Our claim holds:
since $ xw $ is a transfer, after step $ s $ the degree $ d'(w)\leq\maxdeg-2 $ is smaller than that of $ u $ before step $ s $; using $ \maxdeg\geq4 $ we get $ d'(w)\geq\maxdeg-3\geq1 $, since at most two edges incident with $ w $ are removed in step $ s $.
Let $ y\neq w $ be the neighbor of $ u$ or $v $ being selected with degree $ d'(y)\leq \maxdeg-2 $ next.
Recall that step $ s $ removes an edge $ \edge{x,y} $ with $ x\in\{u,v\} $.
If $ y $ is not a node of $ X $, then $ \edge{x,y} $ is an $ F $-edge.
Since $ x,y $ are \m-covered, edge $ \edge{x,y} $ is not a transfer and we have found the other missing debit to one of $ u,v $.
Now assume that $ y $ is a node of $ X $.
At step $ s+1 $ node $ y $ is incident with its \m-edge, maybe with debits and possibly with its \mopt-edge.
Observe that no matter if the \mopt-edge of $ y $ is already removed, node $ y $ has a missing debit since $ d(y)\leq\maxdeg-2 $.
We have found the second missing debit and obtain a contradiction to the
assumption that~$X$ has at most one missing debit.

\ref{lowDebsD}
By parts \ref{lowDebsA}-\ref{lowDebsC} it suffices to prove the case that degrees are bounded by $ \maxdeg=3 $ and~$ d(u)=\maxdeg-1=2 $ holds.
%
%
For singleton~$X$ both endpoints can have at most three debits, since the edge of~$X$
belongs to~$\m$ and hence does not move funds. 
Since $ u $ has a missing debit and by our assumption at most one debit is
missing for~$ X $, we get that $ u $ has exactly one debit, say to $
w_u $. Then $ v $ has exactly two debits, say to~$ w_v$ and~$w_v'$.

Thus,~$u$ can be adjacent to at most one of~$ w_v$ and~$w_v'$, since~$u$ has degree two and is already adjacent to~$v$.
Since the edges~$(u,w_u)$, $(v,w_v)$, and~$(v,w_v')$ are transfers, Definition~\ref{def:transfer} implies that the degrees of~$w_u$, $w_v$ and,~$w_v'$ are each at most~$\maxdeg - 2 = 1$ after~$u$ was matched to~$v$.

If~$u$ is adjacent to neither~$ w_v$ nor~$w_v'$, then~$w_u$, $w_v$, and~$w_v'$ have degree one afterwards, since their degrees were at least the minimum degree of~$d(u) = 2$ before~$(u,v)$ was matched and dropped to at most one afterwards.
Now consider the case that~$w_u$ is also a neighbor of~$v$, say~$w_v = w_u$. Then the degree of~$w_v'$ drops by at most one when the edge~$(u,v)$ is picked by the algorithm, and hence~$w_v'$ has degree one afterwards.
%

In both cases $ w_u,w_v$ or $w_v' $ must be selected and matched by the algorithm in step $ s+1 $, since no other degrees dropped in step $ s $.
A contradiction:
$ w_u,w_v,w_v' $ are assumed to be \mmoptpath endpoints.
%
\end{proof}

We prove Theorem~\ref{thm:threeDeb} for \maxdeg-regular graphs with $ \maxdeg\geq4 $.
A \oneonepath $ X $ has two missing debits by Lemma~\ref{lemma:lowDebs}\ref{lowDebsD}:
the balance is at most $ d_X{-}c_X=d_X{\leq}2(\maxdeg{-}1)-2 $ as required in Eq.~(\ref{totalDebits11}).
An \mmoptpath $ X $ has two missing debits by Lemma~\ref{lemma:lowDebs}\ref{lowDebsC} and $ c_X\geq2(\maxdeg-2) $ credits by Eq.~(\ref{minCredsPath}):
the balance is at most $ d_X-c_X\leq2\w{X}(\maxdeg-2)-2-2(\maxdeg-2) $ as claimed in Eq.~(\ref{totalDebitsmmopt}).

\subsection*{The Proof of Theorem~\ref{thm:threeDeb} for Graphs of Degree At Most~$ \maxdeg=3 $}

By Lemma~\ref{lemma:lowDebs}\ref{lowDebsD}, a \oneonepath $ X $ has a balance of at most $ d_X-c_X=d_X\leq2(\maxdeg-1)-2=2 $ as required in Eq.~(\ref{totalDebits11}).
We prove Eq.~(\ref{totalDebitsmmopt}) for an \mmoptpath~$ X $.
By Eq.~(\ref{minCredsPath}), each path endpoint of~$ X $ receives at least one credit.
If~$X$ has two missing debits, then the balance of~$X$ is at most $ d_X-c_X\leq2\w{X}-4 $, i.e.~$ X $ receives a sufficient amount of~\m-funds.
Hence we assume from now on that~$X$ has at most one missing debit and that each path endpoint of~$ X $ receives exactly one credit.

Lemma~\ref{lemma:lowDebs}\ref{lowDebsA}~and~\ref{lowDebsB} imply that~$X$ has two missing debits if the degree of~$u$ in the creation step of~$X$ is \emph{not}~$\maxdeg-1=2$.
Thus, we focus on the case that the degree of~$u$ is~$ d(u)=2 $ at creation. Since~$u$ is incident with an~\m-edge and an~\mopt-edge, it is not incident with an~$ F $-edge. Therefore, $u$ has a missing debit.
According to our assumption this is the only missing debit of~$X$.

\begin{lemma}
\label{lemma:maxdeg3Graphs}
Assume that the graph has maximum degree~$ \maxdeg=3 $.
If $ X $ is an \mmoptpath with $ d_X=2\w{X}-1 $ debits, then at least $ c_X\geq3 $ credits are given to $ X $.
\end{lemma}
\begin{proof}
To show the statement we assume that~$ c_X<3 $ holds and show a contradiction.

We have already argued that each path endpoint of~$ X $ receives exactly one credit, that the node~$ u $ selected to create~$ X $ has degree~$ d(u)=2 $ at creation, and that~$u$ has the only missing debit of~$X$. 
Then all other \m-covered nodes of $ X $ have exactly
one debit in order to ensure $ d_X=2\w{X}-1 $.
Assume that~$ v $ is matched with~$ u $ in~\m.

First we consider the case of~$ m_X=1 $, i.e.\ that~$ X $ is a \onetwopath. Let~$w_u,w_v$ be the endpoints of~$X$ such that~$(u,w_u), (v,w_v) \in \mopt$. 
Note that this implies~$w_u \neq w_v$.
Then~$v$ has the only debit, say to the augmenting path endpoint~$z_v$.
When~$u$ and~$v$ are matched, the degrees of~$w_u,w_v,z_v$ all drop, and no other degrees drop.
In particular, we claim that~$w_v$ has degree exactly one afterwards, which is the new minimum degree in the graph.
As a consequence, one of~$ w_v,w_u,z_v $ is matched in the next round.
A contradiction is obtained since augmenting path endpoints are never matched.

Now we prove the claim.
When the component is created~$ d(u)=2 $ holds. Since~$u$ is adjacent to~$v$ and~$w_u$, it is not a neighbor of node~$ w_v $. 
Hence~$ w_v $ is incident with at least one~$ F $-edge.
Since only the~\mopt-edge of~$ w_v $ is removed,~$ w_v $ is still incident with at least one~$ F $-edge after creation.
If~$ w_v $ is still incident with at least two~$ F $-edges, then~$ w_v $ receives at least two credits, i.e.\ at least one more than assumed.
So after creation~$ w_v $ is incident with exactly one~$ F $-edge and has degree one.

Next consider the case that the edge~$(u,v)$ belongs to an augmenting path with $ m_X\geq2 $ edges in~\m.
Since~$(u,v) \in \m$ is the first edge picked by the algorithm in~$X$, at least one path endpoint~$ w $ of~$ X $ is still connected with its unique neighbor~$ x $ in~\mopt, after~$ u $ and~$ v $ have been removed with all their incident edges.
We consider the step when~$x$ is matched to its neighbor in~\m, say~$ x' $.
Node $w$ and the recipient of the debit to~$x$ are not yet isolated, since these two nodes are never matched.
Thus  $x$ has degree three before this step.

Moreover, we may assume that~$w$ has degree at most two before this step.
To see this, assume that~$w$'s degree were still~$\maxdeg = 3$ and observe that the two~$ F $-edges incident with~$w$ would become credits, i.e.\ node~$ w $ would get more credits than assumed.

Since~$x'$ has exactly one debit, i.e.~$ x' $ is adjacent to one endpoint of an augmenting path, node $x'$ has degree exactly two before the step.
Consequently, node~$ w $ has degree exactly two as well.

All still present neighbors of~$x,x'$ are endpoints of augmenting paths, and only their degrees drop when~$x$ and~$x'$ are removed from the graph.
If~$x'$ is not adjacent to~$w$, then~$w$'s degree drops by one and to exactly one, hence either~$w$ or another augmenting path endpoint would be picked in the subsequent step.
A contradiction, since an \mmoptpath endpoint is never matched.

Lastly, assume that~$x'$ and~$w$ are adjacent; then~$w$ becomes isolated in that step.
We consider the recipient~$y$ of the debit to~$x$.
Since~$ x' $ is adjacent to~$ w $ and~$ x $ and has degree exactly two, node~$ y $ is not adjacent to~$ x' $.
Therefore the degree of~$ y $ drops by exactly one.
Since~$ y $'s degree drops from at least two before the step to at most~$\maxdeg - 2 = 1$ afterwards by Def.~\ref{def:transfer} we get that~$y$'s degree is exactly one in the subsequent step.
Furthermore, node~$ y $ is now the only degree-1 node and is matched next.
A contradiction since~$ y $ is an \mmoptpath endpoint.
\end{proof}

\subsection{A Guarantee for Graphs with Maximum Degree~$ \maxdeg\geq4 $}
\label{section_mingreedy_guarantee_larger_maxdegs}

In this section we consider graphs with maximum degree~$\maxdeg \geq 4$ and show
a slightly weaker bound. As in Theorem~\ref{thm:threeDeb}, our guarantee holds
for a more general class of greedy algorithms that repeatedly add an edge~$(u,v)$ to the
matching such that either~$u$ or~$v$ has current minimum degree.
For instance, this holds for the deterministic variant of \mingreedy that breaks all ties arbitrarily.
We show the following guarantee.

\begin{theorem}
\mingreedy achieves an approximation ratio of at least \weakerRatio on graphs
with degrees at most \maxdeg.
\end{theorem}
%
%
Why does our bound not match our conjectured approximation ratio
of \desiredRatio?

In Theorem~\ref{thm:threeDeb} we proved a guarantee of~\desiredRatio
for~$\maxdeg$-regular graphs. In order to provide sufficient \m-funds
to \onetwopath{}s, we used that each endpoint of such a short augmenting path
receives the required number of incoming transfers. This was possible, since in
a~$\maxdeg$-regular graph each endpoint has~$\maxdeg - 1$ neighbors besides its
mate in~\mopt.

This property no longer holds for the more general class of graphs we consider
now. 
Thus, we modify our system of transfers.

Recall that an edge in~\m creates a component~$X$ in~$H = (V, \; \m {\cup} \mopt)$ if that edge is the first edge that the algorithm picks in~$X$.

In our new system nodes of an \m-edge~$(u,v)$ that creates some augmenting path do
not transfer \m-funds, i.e.\ the $F$-edges of~$u$ and~$v$ are not debits. As a
consequence, a \onetwopath does not have any debits, because it contains exactly one edge of~\m.
Other than that we do not change the definition of transfers.
Note that \oneonepaths are not augmenting paths and therefore might have debits.

In Lemma~\ref{obs:potentialTransfer} we have shown that each endpoint of an \onetwopath has at least one~$F$-edge, i.e.\ one edge that is eligible for a transfer. Indeed, as we show next, every \onetwopath receives at least one credit.
We prove the following statement for augmenting paths of arbitrary lengths.
\begin{lemma}
\label{lemma_existence_F_edge}
\label{isol12pathEndpointsPrep}
Let~$X$ be an augmenting path. Then at least one of its endpoints is not isolated immediately in the step when the first \m-edge of~$X$ is added to the matching.

Let~$w$ be an endpoint of~$X$ that is not immediately isolated. If~$w$ is not incident with its~\mopt-edge when it is eventually isolated in step~$s$, then~$w$ receives a credit along one of the edges that it was incident with at the beginning of step~$s$.
\end{lemma}
\begin{proof}
If~$X$ contains more than one edge of~\m, then at most one~\mopt-edge can be adjacent to the~\m-edge picked first, and the first part of the lemma follows.
Thus, assume that~$X$ is a \onetwopath.
Assume that both endpoints of~$X$ become isolated upon removing the~\m-edge~$(u,v)$, and note that~$u$ and~$v$ both have degree at least two. Then the endpoints must also have degree two and thus be adjacent to both~$u,v$; recall that the algorithm picks a node of minimum degree. But then the degree of~$u$ and~$v$ is in fact three. This implies that the endpoints must also have degree three, and hence they cannot be isolated by removing~$u$ and~$v$ only.

Now we proof the second part of the lemma. 
%
By assumption of the lemma, at the beginning of step~$s$ the incident edges of endpoint~$w$ are~$F$-edges, and since~$w$ is assumed to be isolated in this step, there can be at most two. 
Assume that~$u'$ is matched to~$v'$ in step~$s$, and recall that their~$F$-edges provide credit for the endpoint~$w$ unless the algorithm creates a new augmenting path by picking the~\m-edge~$(u',v')$.
Thus, it suffices to show that picking~$(u',v')$ does not create an augmenting path. 
Since~$w$ has degree at most two,~$u'$ also has degree at most two.
%
%
%
If $ d(u')=1 $, then $ u' $ does not have both an incident \m-edge and \mopt-edge and hence step $ s $ does not create an \mmoptpath.
Since $ w $ becomes isolated but is not connected to~$ u' $, then~$v'$ is connected to $ w $ by exactly~$F$-edge that provides a credit to~$w$.

On the other hand, if $ d(u')=2 $, then $ d(w)=2 $. Since~$w$ becomes isolated, it must be adjacent to~$u'$ and~$v'$.
If~$u'$ belongs to~$X$, then picking~$u'$ does not create~$X$, since~$X$ was created earlier. Hence~$u'$ and~$v'$ both provide internal credit to~$w$.
Else~$ u' $ is not a node of $ X $. Since~$ u' $ is connected to $ w $, node $ u' $ is not incident both to an \m-edge and an \mopt-edge. But then matching~$u'$ and~$v'$ does not create a new \mmoptpath.
Again both $ u',v' $ transfer a credit to $ w $.
%
\end{proof}
We briefly note that the first part of the lemma also implies that the approximation ratio of~\mingreedy may converge to~$\frac{1}{2}$ asymptotically as shown in Theorem~\ref{theorem_mindegree_randomized}, but will never attain that value exactly.

Since we ensured that a \onetwopath~$X$ has no debits, we claim that the local approximation would be sufficient if we could always provide a second credit.
However, it might be the case that the endpoints of~$X$ have only one neighbor outside~$X$ in total (cf.\ Fig.~\ref{fig:howBadEmerge} for an example).
In this case the second credit to~$X$ will be given via an \emph{indirect transfer}; this is an \m-fund that is not transferred via an~$ F $-edge (see the gray transfer in Fig.~\ref{fig:howBadEmerge}).
To distinguish indirect transfers from~$ F $-edges which move~\m-funds, we also call the latter \emph{direct transfers} from here on.

Interestingly, indirect transfers are also issued by components that are connected to~$X$ via an~$F$-edge.
Also, each indirect transfer originates at a node which is selected when it has degree one.
Consequently, no indirect transfer leaves~$ X $.
So, as desired, no direct or indirect transfers leave~$ X $ at all.

We formally define indirect transfers in Sect.~\ref{section_fake_transfers}. But first we argue that a second incoming transfer to~$ X $ is indeed sufficient to guarantee the claimed local approximations.
\begin{figure}
\centering
\begin{tikzpicture}
\node[] (d) {};
\node[right of=d,very thick] (e) {};
\node[right of=e] (f) {};
\node[right of=f] (g) {$w$};
\node[] at($(g)+(1.5,.55)$)(h) {$v$};
\node[below of=h,very thick] (i) {$u$};
\node[draw=none] (X) at ($(e)+(-1.5,0)$) {$ X $};
\node[draw=none] (1) at ($(e)+(.5,.25)$) {{\tiny 1}};
\node[draw=none] (2) at ($(i)+(+.25,.5)$) {{\tiny 2}};

\draw
(d) edge[opt] (e)
(f) edge[bend left,->,transfer] (d)
(e) edge[mg]node[fill=none,draw=none]{} (f)
(f) edge[opt] (g)
(h) edge[->,transfer] (g)
(h) edge[opt,mg]node[draw=none,fill=none]{} (i)
(f) edge[bend right] (i);
\draw[->,transfer,lightgray,rounded corners,thick]
(i) -- ($ (g)+(1,0) $) -- (g);
;

\end{tikzpicture}
\caption{
	A \onetwopath $ X $ with only one credit from another component:
	\m-edges are picked in order indicated by small numbers (from left to right) selecting fat nodes,
	the gray transfer is \emph{not} an edge of \g and does not effect computations of \mingreedy
}
\label{fig:howBadEmerge}
\end{figure}
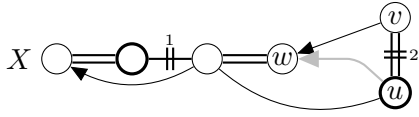

\subsubsection{Optimizing Transferred $\boldsymbol{\m}$-Funds}
Our upper bounds on the balance in Eq.~(\ref{totalDebits11})~and Eq.~(\ref{totalDebitsmmopt}) for \oneonepaths respectively \mmoptpaths are no longer valid for our new system of transfers.
In the sequel we show upper bounds are exactly larger by one unit:
%
A \oneonepath $ X $ can have a balance up to
\begin{align}
\label{newTotalDebitBound11}
d_X-c_X=d_X\leq2(\maxdeg-1)-2+1
\end{align}
and the balance of an \mmoptpath $ X $ with $ \w{X}\geq2 $ is bounded by at most
\begin{align}
\label{newTotalDebitBoundMmopt}
d_X-c_X\leq2\w{X}(\maxdeg-2)-2(\maxdeg-2)-2+1\,.
\end{align}

We have already argued that a~\onetwopath $X$ has a sufficient balance if we can provide credits over two transfers, since no transfers leave~$ X $:
\begin{align}
\label{newTotalDebitBound12}
d_X-c_X\leq 0-2 = -2.
\end{align}

Assume we would move the amount~$ \transferred=\frac{1}{2(2\maxdeg-3)} $ of~\m-funds along each transfer, as we did in case of~$\maxdeg$-regular graphs.
Then a singleton with maximum balance $ d_X=2(\maxdeg-2)+1 $ would have a local approximation ratio of only
\begin{equation*}
\mgapprox_X=\frac{1-\theta(d_X-c_X)}{1}=1-\frac{2(\maxdeg-2)+1}{2(2\maxdeg-3)}=\oneOverTwo\,.
\end{equation*}
In order to obtain a sufficient local approximation for all components, we adjust \transferred.
%
%
What \transferred should we choose?
Parametrized by \transferred, we lower bound local approximation ratios for all components and then optimize \transferred.
%
%
%
\begin{align}
\mbox{\onetwopaths:}&&\mgapprox_X&=\frac{1-\transferred(d_X-c_X)}{2}\stackrel{\mbox{\tiny (\ref{newTotalDebitBound12})}}{\geq}\frac{1+2\transferred}{2}=\oneOverTwo+\transferred\tag{5'}\label{optimizeTheta12}\\
\mbox{\oneonepaths:}&&
\mgapprox_X&
=\frac{1-\transferred(d_X-c_X)}{1} 
\stackrel{\mbox{\tiny (\ref{newTotalDebitBound11})}}{\geq}1-2\theta(\maxdeg-2)-\transferred\tag{6'}\label{optimizeTheta11} \\
\mbox{\mmoptpaths:}&& \mgapprox_X&=\frac{\w{X}-\transferred(d_X-c_X)}{\w{X}+1}\notag\\
&&&\stackrel{\mbox{\tiny (\ref{newTotalDebitBoundMmopt})}}{\geq} \frac{\w{X}-\transferred(1+2\w{X}(\maxdeg-2)-2(\maxdeg-1))}{\w{X}+1}\notag\\
&&&= \frac{\w{X}(1-2\transferred(\maxdeg-2))-\transferred\cdot(1-2(\maxdeg-1))}{\w{X}+1}\notag\\
&&&= 1-2\transferred(\maxdeg-2)+\frac{2\transferred(2\maxdeg-3)-1-\transferred}{\w{X}+1}\notag
\end{align}
Recall that in the last bound for \mmoptpaths we have $ \w{X}\geq2 $.
So if we choose~$ \transferred\geq\frac{1}{2(2\maxdeg-2)} $, then $ 2\transferred(2\maxdeg-3)-1\geq-2\theta $ holds and we can simplify the bound to
\begin{align}
\label{optimizeThetaMmopt}
\mgapprox_X\geq 1-2\transferred(\maxdeg-2)-\frac{3\theta}{3}= 1-2\transferred(\maxdeg-2)-\transferred\,.\tag{7'}
\end{align}
Combining Eq.~(\ref{optimizeTheta12}), Eq.~(\ref{optimizeTheta11}), and Eq.~(\ref{optimizeThetaMmopt}), we set $\oneOverTwo+\transferred=1-2\transferred(\maxdeg-2)-\transferred$ and obtain~\mbox{$\transferred=\frac{1}{2(2\maxdeg-2)}$}.
Hence the local approximation ratio of any component $ X $ is lower bounded by
$$\mgapprox_X\geq\oneOverTwo+\transferred=\oneOverTwo+\frac{1}{2(2\maxdeg-2)}=\frac{2\maxdeg-1}{2(2\maxdeg-2)}=\weakerRatio\,,$$
which proves our claimed performance guarantee for \mingreedy.

\subsubsection{Indirect Transfers}
\label{section_fake_transfers}
To show the balance bounds claimed for \oneonepaths and \mmoptpaths in Eq.~(\ref{newTotalDebitBound11}) respectively Eq.~(\ref{newTotalDebitBoundMmopt}), we first have to develop the definition of indirect transfers.
Therefore we examine the properties of a \onetwopath for which an incoming indirect transfer needs to be added.

Let~$w$ be the endpoint of~$X$ that has receives the only direct credit.
Denote by~$w'$ the other endpoint of~$X$ and assume that~$ X $ is created in step~$ s $.
\begin{enumerate}[label=(c\arabic*),itemindent=1em]
\item\label{12obs1}
Node~$w'$ must be isolated in step~$s$. Otherwise~$w'$ would also receive a direct credit by Lemma~\ref{isol12pathEndpointsPrep} since it loses its~\mopt-edge in step~$s$.
%
%
\item\label{12obs2}
Node~$w$ is not isolated in step~$s$, also by Lemma~\ref{isol12pathEndpointsPrep}.
%
%
\item\label{12obs3}
Let step~$s' > s$ be the step in which~$w$ becomes isolated. At this time~$w$ has exactly one neighbor and receives exactly one direct credit. To see this, assume that~$w$ would be adjacent with both nodes removed in step~$s'$. Then each such~$F$-edge would provide one direct credit, since~$w$'s degree drops from two to zero.
But by assumption~$w$ receives only one direct credit.
%
\end{enumerate}
What are the steps leading to $ w $ becoming isolated eventually?
Recall that an~$F$-edge does \emph{not} provide a direct transfer to~$w$ if the respective neighbor is removed upon creation of a new augmenting path. Thus:
\begin{enumerate}[resume,label=(c\arabic*),itemindent=1em]
\item
\label{12obs4}
\label{12obs4.5}
If~$w$ has degree larger one after its own~\onetwopath is created, then until its degree reaches one it only drops in steps when a new augmenting path is created (and one or both of the nodes matched first is adjacent to~$w$).

Why?
If an adjacent node~$z$ of~$w$ is matched and~$z$ would not belong to the first~\m-edge of an augmenting path, then the edge~$(z,w)$ would provide a second direct credit to~$w$. But we assumed that~$w$ only receives one direct credit. 
%
%
%
%
\end{enumerate}
Assume that in step $ s' $ the algorithm selects node $ u $ and matches it to~$ v $. Note by~\ref{12obs3} node~$w$ has degree one in that step, hence the degree of~$u$ is also one.
Moreover, $v$ is~$w$'s last neighbor, since~$w$ is isolated by removing~$u$ and~$v$.
Therefore,~$(v,w)$ is a direct transfer.
\begin{enumerate}[resume,label=(c\arabic*),itemindent=1em]
\item\label{12obs5}
%
Since~$u$ has degree one when it is matched, all other neighbors (if any) must be matched. Thus,~$u$ has no direct debits to its neighbors.
%
\end{enumerate}
\noindent
Since~$u$ has no direct debits on its own, we add an indirect transfer~$ uw $ from~$ u $ to~$v$.
Moreover, $ vw $ is the direct transfer to $ X $.
%
\begin{definition}
Let $ X $ be a \onetwopath that receives exactly one direct credit $ vw $ to a path endpoint $ w $ of~$ X $.
Let~$u$ be the mate of~$ v $ in~\m. Then we add a transfer $ uw $ and call it an \emph{indirect} transfer.
\end{definition}
In particular, given~$X$ the two nodes that participate in the indirect transfer are uniquely defined.
However, it might be that several \onetwopaths with exactly one direct credit require an indirect transfer from the same node.
Finally we remark:
\begin{enumerate}[resume,label=(c\arabic*),itemindent=1em]
\item\label{12obs7}
No indirect transfer leaves node $ v $, since by construction indirect transfers are added only for its unique \m-neighbor~$ u $.
%
%
\item\label{12obs8}
%
Recall that in our new system of transfers the first \m-edge $ \edge{x,x'} $ picked in an augmenting path does not have any outgoing direct or indirect transfers. 
%
\end{enumerate}


\subsubsection{The Upper Bound on the Balance of \oneonePaths}
For a \oneonepath $ X = (u,v)$ we claimed in Eq.~(\ref{newTotalDebitBound11}) a balance of at most
$$d_X-c_X=d_X\leq2(\maxdeg-1)-2+1\,.$$
Lemma~\ref{lemma:lowDebs}\ref{lowDebsD} gives an upper bound on direct debits of at most $ d_X\leq2(\maxdeg-1)-2 $. Thus, if there is no outgoing indirect transfer, then the claimed bound on~$d_X - c_X$ follows because credits~$c_X$ are always nonnegative.

Assume otherwise and let~$u$ be the node with an outgoing indirect transfer. Then~$u$ has no direct debits by~\ref{12obs5}.
Moreover, $v$ has by the degree constraint at most $ \maxdeg-1 $ direct debits, call them $ vw_1,\dots,vw_{\maxdeg-1} $. Recall from~\ref{12obs7} that~$v$ cannot have outgoing indirect transfers.

%
We show in Lemma~\ref{totalDebitBound11}\ref{totalDebitBound11b} that at most $ \maxdeg-2 $ of the $ w_i $ belong to \onetwopaths which need an incoming indirect transfer, hence at most $ \maxdeg-2 $ indirect transfers leave $ u $.
%
%
So $ d_X\leq(\maxdeg-1)+(\maxdeg-2)=2(\maxdeg-1)-1 $ holds, which gives us the claimed bound on the balance.
%
%
%
%
%
%
%
%
%
%
%
%
We prepare the bound of Lemma~\ref{totalDebitBound11}.

\begin{lemma}
\label{degreesDropToOne}
Let $ \edge{u,v} $ be an \m-edge and $ v $ have direct debits $ vw_1,$ $\dots,$ $vw_n $.
Assume that indirect transfers $ uw_1,\dots,uw_n $ are added.
Before $ u,v $ are matched, the degrees of all $ w_i $ drop to $ d(w_i)=1 $ in the same step \sdown.
\end{lemma}
\begin{proof}
By definition of indirect transfers, all recipients~$w_1,\ldots,w_n$ are isolated when~$u$ and~$v$ are matched with each other. We denote this step by~$s$.
By~\ref{12obs3} each node~$w_i$ has degree one at the beginning of step~$s$.
By~\ref{12obs4.5}, for all~$ i $ the degree of~$ w_i $ dropped to~$ d(w_i)=1 $ in a step $ \sdown^i $ with~$\sdown^i < s$, when an \mmoptpath was created by picking an \m-edge $ e_i $.
Recall that step $ \sdown^i $ selected a node of degree at least two.

Assume that there is a step when we have $ d(w_j)=1 $ and $ d(w_k)\geq2 $ for $ j\neq k $.
Until~$ w_j $ is isolated in step $ s $, \mingreedy picks only nodes of degree one and hence step~$ \sdown^k $ does not happen until after step $ s $.
Therefore $ d(w_k)\geq2 $ holds when $ w_j $ is isolated, a contradiction since all $ w_i $ are isolated in step~$s$ when all $ d(w_i)=1 $.
So all $ d(w_i) $ are decreased to $ d(w_i)=1 $ by the same step $ \sdown=\sdown^1=\dots=\sdown^k $ creating an \mmoptpath by picking an \m-edge $ e_1=\dots=e_k $.
\end{proof}

\begin{lemma}
\label{totalDebitBound11}
Let $ \edge{u,v} $ be an \m-edge and assume that indirect transfers $ uw_1,$ $\dots,$ $uw_n $ are added.
Denote by~\sdown the step when the degrees of all $ w_i $ drop to $ d(w_i)=1 $. Then:
\begin{enumerate}[topsep=0mm,noitemsep,label=\alph*)]
\item\label{totalDebitBound11a}
In step~\sdown for every~$w_i$ an incident~$F$-edge is removed.
\item\label{totalDebitBound11b}
It holds that $ n\leq\maxdeg-2 $.
\item\label{totalDebitBound11c}
If $ n=\maxdeg-2 $, then step \sdown removes exactly $ \maxdeg-2 $ edges of $ F $.
\end{enumerate}
\end{lemma}
\begin{proof}
We prove~\ref{totalDebitBound11a}.
%
By~\ref{12obs4.5} step \sdown creates an \mmoptpath $ X $, since at step \sdown
the degrees of the $ w_i $ drop to $ d(w_i)=1 $.
First consider the case that~$w_i$ is not an endpoint of the augmenting path~$X$.
Then the edge that connects~$w_i$ to its neighbor in~$X$ is an~$F$-edge.
Now let~$w_i$ be an endpoint of~$X$.
Since in step~\sdown the algorithm selects a node of degree at least two, we have $ d(w_i)\geq 2 $ at that time.
But $ w_i $ is incident to at most one edge of $ X $, its~\mopt-edge, and hence to at least one $ F $-edge.

We prove~\ref{totalDebitBound11b}~and~\ref{totalDebitBound11c}.
W.l.o.g. step $ \sdown $ selects $ \udown $ and matches $ \udown$ with $\vdown $.
Step $ \sdown $ selects $\udown $ when $ 2\leq d(\udown) \leq 3$ since an \mmoptpath is created and the degrees of the $ w_i $ drop to $ d(w_i)=1 $ from at most $ d(w_i)\leq3 $.

Assume that $ d(\udown)=2 $.
Node $ \udown $ is incident only to its \m- and \mopt-edge.
So the~$ n $ distinct $ F $-edges being removed by~\ref{totalDebitBound11a} are incident to $ \vdown $.
Since $ \vdown $ is also incident to its \m- and \mopt-edge, we get $ n\leq d(\vdown)-2\leq\maxdeg-2 $.
Exactly $ n $ many $ F $-edges, namely~$ \edge{\vdown,w_1},\dots,\edge{\vdown,w_n} $, are removed, which holds in particular for $ n=\maxdeg-2 $.

Assume that $ d(\udown)=3 $.
Let $ X_i $ be the \onetwopath of $ w_i $.
Step $ \sdown $ does not pick the~\m-edge of one of the $ X_i $, since otherwise an \mopt-endpoint $ w' $ of the created \onetwopath would have to get isolated by~\ref{12obs1}, implying the contradiction $ d(w')\leq2<d(\udown) $.
So $ X\neq X_i $ for all $ i $.
At step \sdown we have $ d(w_i)\geq d(\udown)=3 $ for all $ i $.
Each $ w_i $ is connected to each of~$ \udown,\vdown $, since otherwise the degree of $ w_i $ could not drop to $ d(w_i)=1 $.
Hence at step~$ \sdown $ each of~$ \udown,\vdown $ is incident to $ n $ edges of $ F $.
Again we get $ n\leq\maxdeg-2 $, since each of $ \udown,\vdown $ is also incident to its \m-edge and \mopt-edge.
Assume that $ n=\maxdeg-2 $ holds.
Using $ \maxdeg\geq4 $ we get $ d(\udown)=2+\maxdeg-2>3=d(\udown) $, a contradiction.
(Hence if $ n=\maxdeg-2 $, then case~$ d(\udown)=2 $ applies, where exactly $ n $ many $ F $-edges are removed.)
\end{proof}

\subsubsection{The Upper Bound on the Balance of \mmoptPaths}
For an \mmoptpath $ X $ with $ \w{X}\geq2 $ we claim in Eq.~(\ref{newTotalDebitBoundMmopt}) a balance of at most
\begin{align*}
d_X-c_X 
& \leq 2 \cdot (\w{X}-1) \cdot (\maxdeg-2)-1\,.
\end{align*}
Recall that the nodes of the~\m-edge that was picked upon creation of~$X$ have no outgoing transfers at all (cp.~\ref{12obs8}).
For the other~$m_X - 1$ edges both nodes are incident to at most~$\maxdeg - 2$ edges in~$F$ each, and these edges could move~\m-funds out of~$X$. In particular, we have already argued that indirect transfers do not increase the overall number of transfers out of an~\m-edge: If $(u,v) \in \m$ and $ u $ has outgoing indirect transfers, then by~\ref{12obs5} and~\ref{12obs7} node $ u $ has no direct debits and~$ v $ has no outgoing indirect transfers.
In particular, the number of indirect transfers leaving~$u$ is bounded above the number of direct debits to~$v$, which in turn is at most~$\maxdeg - 2$.

If one of the~\m-covered nodes that was not matched upon creation of~$X$ has less than~$\maxdeg - 2$ outgoing transfers, then the upper bound on the balance claimed in Eq.~(\ref{newTotalDebitBoundMmopt}) is implied. Hence we assume from now on that~$X$ has~$ 2\cdot (\w{X}-1) \cdot (\maxdeg-2) $ outgoing transfers. Then the claim follows from the next lemma.
%
\begin{lemma}
Let $ X $ be an  \mmoptpath  with $ \w{X}\geq2 $.
If $ X $ has its maximum number of $ d_X=2 \cdot (\w{X}-1) \cdot (\maxdeg-2) $ outgoing transfers, then $ X $ has at least one direct credit.
\end{lemma}
\begin{proof}
Let~$w$ denote an endpoint of~$X$ that is not isolated upon creation.
Lemma~\ref{isol12pathEndpointsPrep} states that such a node~$w$ exists. Moreover, Lemma~\ref{isol12pathEndpointsPrep} also states that~$X$ receives the desired~\m-fund via a direct transfer to~$w$, if~$w$ is not incident to its~\mopt-neighbor~$w^\ast$ in the step~$s$ when~$w$ is isolated.
Thus, we assume in the sequel that~$w$ becomes isolated when the~\m-edge~$(x,w^\ast)$ is picked.
Note that~$(x,w^\ast)$ is not the first edge picked in~$X$ according to the choice of~$w$. If~$w$ is adjacent to~$x$, then the~$ F $-edge~$(x,w)$ provides the desired direct transfer from~$x$ to~$w$. 
	
But if~$w$ is not adjacent to~$x$, then~$w$ has degree one in step~$s$. Hence~$x$ must also have degree one, because~$w^\ast$ is still adjacent to~$w$ and~$x$, and thus~$x$ is a node of minimum degree. Then~$x$ has no direct debits, but it might have indirect transfers to path endpoints that are adjacent to~$w^\ast$. Since we assumed that~$X$ has its maximum number of outgoing transfers, node~$x$ has exactly~$\maxdeg - 2$ outgoing indirect transfers; call them $ xw_1,\dots,xw_{\maxdeg-2} $.
If there is a~$w_k$ that belongs to~$X$, then~$w_k$ receives a direct transfer from~$w^\ast$, which provides the desired credit to~$X$.
Thus, we assume from now on that all~$w_i$ do not belong to~$X$.
It is crucial to note that~$w \neq w_i$ for all~$i$ because~$w$ is in~$X$.

%
By Lemma~\ref{degreesDropToOne}, the degrees of all $ w_i $ drop to $ d(w_i)=1 $ in a step $\sdown<s$, and by~\ref{12obs4} step~$\sdown$ creates an \mmoptpath.
So at step $ s $ we have $ d(w)=d(w_1)=\dots=d(w_{\maxdeg-2})=1 $.
Since the degree of $ w $ in the input~\g is at least two by Lemma~\ref{obs:potentialTransfer}, there is a step~$\sdown'$ with~$\sdown'<s $ when the degree of $ w $ is decreased to $ d(w)=1 $ by removing an incident~$F$-edge.

\begin{itemize}
\item
Assume that $ \sdown'=\sdown $.
Since step $ \sdown' $ does not remove the \mopt-edge $ \edge{w,w^\ast} $, which is removed later in step~$ s $, step~$ \sdown' $ removes an edge of $ F $ incident to $ w $.
Also, by Lemma~\ref{totalDebitBound11}\ref{totalDebitBound11a}, step $ \sdown' $ removes~$ \maxdeg-2 $ edges of $ F $ incident to the endpoints~$ w_i$.
But since $ w\neq w_i $, for all $ i $, step~$ \sdown' $ removes at least $ \maxdeg-1 $ many edges of $ F $.
A contradiction to Lemma~\ref{totalDebitBound11}\ref{totalDebitBound11c}.
\item
Assume that $ \sdown'\neq\sdown $.
Since in step $ \sdown' $ the degree of $ w $ drops to $ d(w)=1 $ and step $ \sdown $ selects a node of degree at least 2, step $\sdown$ happens before step $ \sdown' $.
Since step \sdown decreases the degrees of all $ w_i $ to $ d(w_i)=1 $, step $ \sdown' $ is a degree-1 step, i.e.\ a step when a node of degree one is picked.
Consequently, step $ \sdown' $ does not create a new augmenting path, especially not~$ X $.
Also, we claim that step $ \sdown' $ does not pick any other \m-edge of $ X $:
To see this, recall that step $ \sdown' $ does not remove \edge{w,w^\ast}, which is removed by step $ s $. Nor does $ \sdown' $ remove an edge of $ F $ incident to $ w $, since such an edge would imply an `internal' direct transfer.

Hence step $ \sdown' $ picks the \m-edge of a component other than $ X $. We already argued that $ \sdown' $ does not create an \mmoptpath, thus there is a direct credit to $ w $ coming from the \m-edge picked by step $ \sdown' $.
\end{itemize}
\end{proof}

%% file: inapproxTheModel.tex
In order to study the limitations of greedy algorithms, we utilize the model of
priority algorithms. The crucial idea is to regard the input graph~$\cal I$ as a
collection of data items, where in the \emph{vertex model} a data item
corresponds to a node~$u$ in~$\cal I$ with its respective
neighbors~$v_1,\ldots,v_d$.
The data item is denoted by~$\langle u;v_1,\ldots,v_d\rangle$.
An \emph{adaptive priority algorithm} chooses an ordering~$\pi$ on the set of all
possible data items, i.e.\ without actually looking at~$\cal I$.
Then it receives the first data item~$d$ of~$\cal I$ w.r.t.~$\pi$ such that the
node~$u$ of~$d$ is still matchable, i.e.\ neither already matched nor isolated.
Now the algorithm has to make an irrevocable decision:
Either it chooses a matchable neighbor~$v$ and matches~$u$ to~$v$,
or~$u$ becomes isolated and cannot be matched afterwards. Then the algorithm
iterates until no matchable nodes are left. 
Such an algorithm is called \emph{greedy} if it may not choose to isolate~$u$.
All orderings and decisions are performed \emph{deterministically}, but may
take into account all information about~$\cal I$ gathered so far. In particular,
priority algorithms are not resource-bounded.

The adaptive priority game is a convenient way to present inapproximability
results by turning the above definition into a game between the algorithm~$A$
and an adversary~$B$ (see~\cite{bblm10} for a primer).
Initially,~$B$ selects a private graph~$\cal I$, then the game proceeds in
rounds until no matchable nodes are left in~$\cal I$: In each round~$A$ submits
an ordering~$\pi$ on all possible data items and receives the first data item of
a matchable node~$u$ from~$B$.
Then~$A$ makes an irrevocable decision for~$u$, thereby ending the round.

%% file: inapproxFullyRandomized.tex
\label{section_nongreedy_adaptive_algorithms_matching}
\label{section_randomized_priority}
Angelopoulos and Borodin~\cite{ab10} introduced \emph{fully
randomized priority algorithms}: These algorithms proceed like adaptive priority
algorithms, but may utilize randomness when determining an ordering of the data
items and making decisions. 

%
The class of fully randomized algorithms is quite comprehensive, as it contains
for instance the algorithms \greedy, \mingreedy, \mrg, \ranking, and the
\karpsipser algorithm~\cite{ks81}.
An exception is~\MDS that we studied in Theorem~\ref{theorem_mindegree_varianten}.
%

First we study deterministic priority algorithms and show an inapproximability
bound for adaptive priority algorithms. The underlying construction will provide
the basis for our investigation of fully randomized priority algorithms.
\begin{theorem}
\label{theorem_adaptive_nongreedy}
No adaptive priority algorithm, whether greedy or not, achieves approximation
ratio better than~$\frac{2}{3}$ in the vertex model. 

The bound holds for graphs with maximum degree three, and hence the
deterministic \MDeg is an optimal adaptive priority algorithm for these graphs.
\end{theorem}
\begin{proof}
Given such a deterministic algorithm $ A $, we consider the two input graphs
in Fig.~\ref{figure_nongreedy_adaptive_algos_deg2}
and Fig.~\ref{figure_nongreedy_adaptive_algos_deg3}, both with perfect matching.
When the game starts,~$A$ submits an ordering~$\pi_1$ on the set of all data
items: Depending on~$\pi_1$, the first data item~$d$ gives a node of
degree two or three; let~$d$ be $ \langle u;v,w\rangle $ or $ \langle
u;v,w,z\rangle $.
\begin{figure}[htbp!]
\centering
\begin{minipage}[t]{.45\textwidth}
\begin{center}
\begin{tikzpicture}
\node[,very thick](A){$u$};
\node[below=.5cm of A](B){$w$};
\node[draw=none, below=0.025cm of A](X){};
\node[right=7mm of X](C){$v$};
\node[right=8mm of C](D){$z$};
\node[right=3cm of A](E){$b$};
\node[right=3cm of B](F){$c$};

\draw[opt] (A) -- (B);
\draw[mg] (A) -- (C);
\draw[] (C) -- (B);
\draw[opt] (C) -- (D);
\draw[] (D) -- (E);
\draw[] (D) -- (F);
\draw[opt] (E) -- (F);

\end{tikzpicture}
\caption[Adaptive Priority Algorithms: The Case deg($u$) = 2]{Algorithm $ A $ receives $ \langle u; v,w\rangle $}
\label{figure_nongreedy_adaptive_algos_deg2}
\end{center}
\end{minipage}
\hspace{0.5cm}
\begin{minipage}[t]{.45\textwidth}
\begin{center}
\begin{tikzpicture}
\node[draw](A){$v$};
\node[below=.5cm of A](B){$w$};
\node[draw=none, below=0.025cm of A](X){};
\node[right=7mm of X,very thick](C){$u$};
\node[right=8mm of C](D){$z$};
\node[right=3cm of A](E){$b$};
\node[right=3cm of B](F){$c$};

\draw[opt] (A) -- (B);
\draw[mg] (A) -- (C);
\draw[] (C) -- (B);
\draw[opt] (C) -- (D);
\draw[] (D) -- (E);
\draw[] (D) -- (F);
\draw[opt] (E) -- (F);

\end{tikzpicture}
\caption[Adaptive Priority Algorithms: The Case deg($u$) = 3]{Algorithm $ A $ receives $ \langle u; v,w,z\rangle $}
\label{figure_nongreedy_adaptive_algos_deg3}
\end{center}
\end{minipage}
\end{figure}
In both cases, if $ A $ decides \emph{not} match $ u $, then $ A $ will not
obtain a matching larger than two edges, and approximation ratio at most
\twoOverThree follows.

How does the adversary proceed if $ A $ matches $ u $? $ A $ has no knowledge
about the graphs and we may assume that~$u$ matches~$v$.
In particular, both graphs are indistinguishable for~$A$.
If~$u$ has degree two, then the input is the graph in
Fig.~\ref{figure_nongreedy_adaptive_algos_deg2}. Otherwise it is the one in
Fig.~\ref{figure_nongreedy_adaptive_algos_deg3}.
%
Thus, after matching~$(u,v)$,~$A$ can match only one more edge, therefore the
claimed bound follows.
The inapproximability bound matches the guarantee for the deterministic variant
of \MDeg given in Theorem~\ref{thm:threeDeb} for graphs of maximum degree three.
\end{proof}
%
%
%
\begin{theorem}
\label{theorem_fully_randomized_algorithms_maximum_matching}
No \emph{fully randomized} priority algorithm can achieve an expected
approximation ratio better than~$\frac{5}{6}$ for the vertex model.
\end{theorem}
\begin{proof}
We apply Yao's Minimax Principle \cite{yao83}.
We have to construct a hard distribution over
input instances and to analyze the best deterministic algorithm (that knows the
distribution).
As distribution we take all the graphs corresponding to the permutations of
the node labels of Fig.~\ref{figure_nongreedy_adaptive_algos_deg2}.
%
We will consider only mistakes made in the first round and assume that the
algorithm proceeds optimally afterwards.

First note that if the algorithm decides to isolate the node given in the first
round, it cannot obtain a matching larger than two. Thus, we may assume that the
first node is matched.
Furthermore, if the first matching is non-optimal, the algorithm again obtains
at most two edges and has approximation ratio at most~$\frac{2}{3}$.
On the other hand, if the first matching is optimal, a maximum matching can be
obtained.

Since we are in the first round, the algorithm has no information which neighbor
is the optimal choice, and any neighbor is the optimal mate with same
probability because we picked a labeling of the nodes uniformly at random.
Thus, the best strategy for the algorithm is to request no degree three prior to
degree two, since the probability of matching the first node optimally decreases
with its degree, and the bound follows because a node of degree two is matched
optimally with probability~$\frac{1}{2}$.
The bound is still valid if the number of nodes, the number of edges, and all
degrees are revealed in advance by the adversary.
\end{proof}
We compare our inapproximability result for fully randomized algorithms to the
bounds obtained by Goel and Tripathi~\cite{gt12}.
On the one hand, they studied randomized greedy algorithms in the
\emph{oblivious query commit model}.
%
%
In this model edges are not revealed to the algorithm; the only way to figure
out whether a particular edge exists is to probe the pair of its endpoints. If
an edge is found whose endpoints are both not matched yet, it must be added to the matching
irrevocably. 
In this case both nodes are removed from the graph.

Note that it is impossible to design an algorithm that chooses nodes depending
on their degrees, and in particular the algorithm cannot select a node of
degree one. That's why none of these randomized algorithms achieves an  expected
approximation ratio better than~$\frac{19}{24} \approx 0.792$ on a triangle
with a single edge attached.

On the other hand, Goel and Tripathi consider the more restricted class of \emph{vertex iterative} algorithms.
%
%
A vertex iterative algorithm picks randomly a vertex, say~$u$, in each round and
then may scan (a subset of) the other vertices, one after another, to check
whether they are adjacent to~$u$. As required by the oblivious query commit model, whenever an edge is
found, it is added to the matching. Moreover, before every probe the algorithm
may choose to isolate the currently inspected vertex~$u$ irrevocably and thereby
skip to the next round.
%

Goel and Tripathi show that no vertex iterative algorithm obtains an expected
approximation ratio better than~$\frac{3}{4}$ on the graph of Dyer and
Frieze~\cite{df91}.

The~\AlgFrieze algorithm and \AlgKarp are prominent representatives of vertex iterative algorithms.
\mingreedy, however, cannot be implemented in the oblivious query commit model.

We point out that the class of fully randomized adaptive priority algorithms in the vertex model contains all randomized algorithms in the oblivious query commit model.
\begin{theorem}
Every algorithm in the oblivious query commit model (and hence every vertex iterative algorithm) can be implemented as fully randomized priority algorithm in the vertex model.

On the other hand, the fully randomized priority algorithm \mingreedy cannot be implemented in the oblivious query commit model.
\end{theorem}
%
\begin{proof}
We have already pointed out that the oblivious query commit model does not allow the algorithm to select edges depending on the degrees of their nodes. In fact, the lower bound given in~\cite{gt12} relies on this observation. In the sequel we show the first claim of the lemma.

Recall from the definition of vertex iterative algorithms that each such randomized algorithm can be implemented in the oblivious query commit model.
Thus, it suffices to demonstrate that an algorithm~$ A $ in the oblivious query commit model can be simulated by some fully randomized greedy algorithm~$ B $. 
Let~$G$ be the input on which both algorithms are run. Since~$G$ is unknown to both algorithms, we assume for the sake of convenience that both algorithms are provided a set of possible node identifiers in advance. In particular, we assume that an upper bound on the number of nodes is common knowledge.

We assume that all random bits that~$A$ uses are drawn in advance. Thus, given these random bits~$B$ can simulate~$A$ on~$G$ deterministically.
Algorithm~$B$ determines the first query of~$A$, say for~$(u,v)$.
$B$ simulates the first query by giving highest priority to all possible data items of node~$u$ that contain~$v$ as neighbor. Since exactly one of these data items is consistent with~$G$, their ordering w.r.t.\ each other is irrelevant. Let~$\pi_1$ be the ordering submitted by~$B$ in the first round.
If~$(u,v)$ exists, $A$ adds this edge to its matching, and so does~$B$.

If~$(u,v)$ does not exist, then~$A$ may query probabilistically for another edge, say~$(x,y)$. 
But~$B$ will receive the first data item according to~$\pi_1$ that exists in the graph. In particular, the priority algorithm may only change its ordering of data items after it has received a data item; since there is no such data item for~$u$ in~$G$, $B$ may take no action at this time.
Fortunately,~$B$ has access to the description of~$A$ and its pool of random bits, hence~$B$ can determine a priori which edge~$A$ would query next if~$(u,v)$ does not exist. Therefore, in the first round $B$ plans ahead and enumerates all data items for node~$x$ that contain~$y$ as neighbor after the prefix of the data items for~$u$ in~$\pi_1$. We iterate this process.

Assume that~$A$ eventually finds some edge~$(u',v')$ and adds it to the matching. Then~$B$ receives the data item for~$u'$ and picks the same edge. The gist is that~$B$ can infer the same information from the data item for~$u'$ that~$A$ has gathered: The position of the data item in~$\pi_1$ implies that no edge of higher priority exists in~$G$. Moreover, the data item of~$u'$ contains at least the information that~$(u',v')$ exists, and perhaps additional information about the neighborhood of~$u'$.

Thus, $B$ has always at least the same knowledge about~$G$ and hence can simulate~$A$ subsequently.
Here it is crucial to note that according to the respective definitions of their models, $A$ and~$B$ may only query data items of nodes that have not been matched yet.

\noindent{}Once~$B$ has decided the first data item, it determines~$\pi_2$, $\pi_3$ and so on analogously.
\end{proof}	

%% file: inapproxGreedyBoundedDegree.tex
\subsection{Greedy Adaptive Priority Algorithms and Degree Bounded Graphs}
\label{section_greedy_priority}
The inapproximability bounds given in
Sect.~\ref{section_nongreedy_adaptive_algorithms_matching}
rely on graphs with maximum degree three.
How do priority algorithms perform when applied to arbitrary graphs?

Recall that \emph{greedy} adaptive priority algorithms do not have the option to
isolate the node given in the current data item; they must add an edge to their
matching in each round.
We show that every such algorithm has approximation ratio at most
$ \frac{1}{2}+\varepsilon $ for any $ \varepsilon>0 $. Thus, randomness seems
essential for greedy algorithms in order to achieve a non-trivial guarantee.
%

\begin{figure}[htbp!]
\centering
\begin{minipage}[t]{.32\textwidth}
\centering
\input{figures/trivialComponent}%
\caption{
A connected component of \g:
Gray edges are unknown to algorithm $ A $
}
\label{unknownHigh_}
\end{minipage}
~
\begin{minipage}[t]{.65\textwidth}
\centering
\input{figures/theHardInstance}
\caption{
The full construction:
Gray nodes and edges are unknown to $ A $,
frontier nodes are drawn bold,
the dashed edge is an example for edge \edge{m_i=v,r_j} in a type 3 round
}
\label{known_}
\end{minipage}
\end{figure}

\begin{theorem}
\label{apvUpperBounds}
Let $ A $ be a greedy adaptive priority algorithm.
There is a graph with maximum degree at most~\maxdeg, for which the
approximation ratio of~$ A $ is not better than~\desiredRatio.
\end{theorem}

\begin{proof}
%
The adversary creates the input graph on-the-fly during the adaptive priority
game; this is legal as long as the adversary ensures that the final graph is
consistent with the revealed data items, since the algorithm works
deterministically.
We call a node \emph{known}, if it was contained in a data item shown in a
previous round, and \emph{unknown} otherwise. The gist of the construction is
that all known nodes are either already matched or isolated.

The game between $ A $ and the adversary $ B $ consists of two
phases:
The \emph{regular game} lasts for $ s=\maxdeg-3 $ rounds, the \emph{endgame} has
two rounds.
We consider round~$i \geq 1$ of the regular game and
assume that~$B$ returns data item~$a_i$ that belongs to one of the following
three types; $B$ will give no other data item to~$A$.

%
%

\noindent\textbf{Type 1: $ a_i=\langle v; v_{1}, \dots, v_{d} \rangle $ with
$ 3\leq d\leq \maxdeg $ and all nodes are unknown.} Thus, the nodes are
indistinguishable and we may assume that~$A$ matches~$ v $ to~$ v_1 $.
$ B $ constructs a separate connected component $ C $ (cf. Fig.~\ref{unknownHigh_}):
The optimum~\mopt contains two edges $ \edge{v,v_2},\edge{v_1,v_3}$ in~$C$,
whereas $ A $ adds only the edge~$\edge{v,v_1} $ to its matching~\m.
Observe that all (data items of) nodes in~$ C $ belong to type~1 or~2 before the
current round and are either matched or isolated afterwards.


\noindent\textbf{Type 2: $ a_i=\langle v; v_{1}, v_{2} \rangle$, and all nodes
are unknown.} 
Assume that $ A $ matches $ v$ with $v_1 $.~$ B $ constructs a triangle $
\{l_i,m_i,r_i\} $ with $ l_i{=}v_2,m_i{=}v,r_i{=}v_1 $ and an edge $
\edge{r_i,u_i} $ with a new node $ u_i $, which connects the triangle to the
unknown \emph{center} (cf. Fig.~\ref{known_}). The center will connect triangles
created for data items of types~2 and~3.
Again \mopt is extended by two edges~$ \edge{l_i,m_i},\edge{r_i,u_i} $,
whereas only~$\edge{m_i,r_i} $ is added to~\m.
To verify the legality, we observe that before $ m_i,r_i $ are matched, (the
data items of) nodes $ m_i,l_i $ are of type 2 and $ r_i,u_i $ are of type 1;
after matching $m_i,r_i$, nodes $ l_i,m_i,r_i $ are isolated and $ u_i $ turns
into a type 3 node.

\noindent\textbf{Type 3: $ a_i=\langle v; v_{1}, v_{2},v_{3} \rangle $, where $
v,v_1,v_2 $ are unknown and $ v_3 $ is known.}
Node $ v_3 $ occurred in a data item presented previously, and in particular
must be a neighbor of some node~$r_j$ (with~$j < i$) by construction.
%
%
But then, does $ v=u_j $ hold? Not necessarily, since $ B $ may introduce
further neighbors of~$ r_j $, since~$r_j$ was matched on its first appearance
and hence its data item is never presented to the algorithm.
%
Since $ v_3=r_j $ is already matched and $ v_1,v_2 $ are both unknown, we may
assume that $ A $ matches $ v $ to~$ v_1 $.
Again~$ B $ creates a triangle $ \{l_i=v_2,m_i=v,r_i=v_1\} $
and an additional edge \edge{r_i,u_i} with a new node $ u_i $. Moreover,~$ B $
also inserts the edge $\edge{m_i{=}v,r_j}$ to preserve consistency.
\mopt (resp.,~\m) is extended by~$ \edge{l_i,m_i},\edge{r_i,u_i} $ (resp., $
\edge{m_i,r_i} $).
Before $ m_i,r_i $ are matched, node $ l_i $ is of type 2, nodes $ r_i,u_i $ are
of type 1 and $ m_i{=}v $ is of type 3.
After matching $m_i,r_i$, nodes $ l_i,m_i,r_i $ are isolated and $ u_i $ turns
into type 3. $ u_j $ is still of type 3.

The regular game ends after round~$s = \maxdeg - 3$.
%
We consider the graph created by~$B$ (cf. Fig.~\ref{known_}):
$a,b$ are of type~1 (for a specific value of~$d$), $b,c$ of type~2, and
the~$u$-nodes of type~3.
The edges matched in the endgame are all incident in the center: $B$
enforces that $A$ matches only two edges, whereas the optimum obtains
three. Summing up, we have $ |\m|=s+2=\maxdeg-1 $ and $ |\mopt|=2s+3
=2\maxdeg-3$, and the claim follows.
%
%
$B$ asserts that the edge~$(a,b)$ is matched in round~$\maxdeg - 2$;
afterwards~$c$ can be matched to some neighbor, leaving all other nodes
isolated.
We distinguish the following types for data item~$ a_{\maxdeg-2} $:

\noindent\textbf{I) $ a_{\maxdeg-2}=\langle v; v_{1}, \dots, v_{d} \rangle
$ is of type~1.} Since no nodes in $ a_{\maxdeg-2} $ is known, we assume
that $ A $ matches~$ v $ to~$ v_1 $. $ B $ chooses $ v =a$, $ v_1=b $, and $ v_2,\dots,v_d $ as
the remaining neighbors of $ a $.

\noindent\textbf{II) $ a_{\maxdeg-2}=\langle v; v_{1}, v_{2} \rangle $ is of
type 2.}
Again we may assume that $(v,v_1)$ is matched, hence~$ B $ chooses $ v =b$, $
v_1=a $, and $ v_2=d $.

\noindent\textbf{III) $ a_{\maxdeg-2}=\langle v; v_{1}, v_{2},v_{3} \rangle $ is
of type~3.}
As above, the known node~$ v_3 $ is some matched node $r_j$, $j<\maxdeg-2 $, and
we may assume that $(v,v_1)$ is matched by~$ A $.
The adversary chooses $v_3 = r_j$, $ v{=}b$, $v_1{=}a$, and~$v_2{=}d$;
therefore, $B$ creates the edge~$(v_3,b)$ (not present in Fig.~\ref{known_}).

Concludingly, we verify that no node has degree larger than~$\maxdeg$:
Nodes in type~1 components have degree at most \maxdeg by definition of the
component.
The degree of $ a$ and $c $ is at most~$ \maxdeg = 3+s $, since each
round of the regular game adds at most one~$u$-neighbor to both.
At most $ s-1 $ neighbors are added to an $ r $-node during the regular game, at most one neighbor is added in (the first step of) the endgame, hence degrees of $ r $-nodes are at most \maxdeg as well.
All other nodes have degree at most three.
%
\end{proof}

%% file: figures/trivialComponent.tex
\begin{tikzpicture}
\input{tikzStyles}
\tikzstyle{every node} = [circle, fill=white,draw=black,minimum size=13pt,inner sep=0pt];

\node[,very thick] (v) {$ v $};
\node[] (w) at ($(v) + (3,0)$) {$ v_1 $};
\node[] (x) at ($(v)+(0,-2)$){$ v_2 $};
\node[] (z) at ($(x)+(1,0)$){$ v_3 $};
\node[] (y) at ($(z)+(2,0)$) {$ v_d $};
\node[draw=none] at ($(y)+(-1,0)$) {$ \dots $};

\draw
(w) edge[comp,bend right=15] (x)
(w) edge[comp,bend right=15] (z)
(w) edge[comp] (y)
(w) edge[,bend right=15] (x)
(v) edge[mg] (w)
(v) edge[opt] (x)
(v) edge[,bend left=15] (z)
(w) edge[compopt,bend right=15] (z)
(w) edge[] (y)
(v) edge[,bend left=15] (y)
;

\draw[white] ($(v)+(-.4,+.4)$) rectangle ($(y)+(+.4,-.4)$);
\end{tikzpicture}

%% file: figures/theHardInstance.tex
\begin{tikzpicture}
\input{tikzStyles}

\node[,very thick] (v') {$ m_1 $};
\node[] (w') at ($(v') + (1,0)$) {$ \boldsymbol{r_1} $};
\node[left of=v'] (x') {$ l_1 $};
\node[,right of=w',,fill=lightgray] (g') {$ u_1 $};

\node[,very thick] (v) at ($(v') + (0,-1)$) {$ m_k $};
\node[] (w) at ($(v) + (1,0)$) {$ \boldsymbol{r_k} $};
\node[left of=v] (x) {$ l_k $};
\node[right of=w,,fill=lightgray] (g) {$ u_k $};

\node[,very thick] (v''') at ($(v') + (0,-2)$) {$ m_i $};
\node[] (w''') at ($(v''') + (1,0)$) {$ \boldsymbol{r_i} $};
\node[left of=v'''] (x''') {$ l_i $};
\node[right of=w''',fill=lightgray] (g''') {$ u_i $};

\node[mylabel] (vdots') at ($(x')+(0,-.4)$) {$\vdots$};
\node[mylabel] (vdots') at ($(g')+(-.5,-.4)$) {$\vdots$};

\node[,fill=lightgray] at ($(g') + (2,0)$) (c) {$ a $};
\node[fill=lightgray]  at ($(c) + (0,-2)$) (c') {$ c $};
\node[fill=lightgray,,very thick] at ($(c)+(2,0)$) (c'') {$ b $};
\node[fill=lightgray] (d) at ($(c)+(0,-1)$) {};
\node[fill=lightgray,right of=d] (d') {};
\node[fill=lightgray] (d'') at ($(c')+(2,0)$) {$ d $};

\draw[lightgray,thick,rounded corners=6pt] ($(c)+(-.4,.4)$) rectangle ($(d'')+(.4,-.4)$);
\node[draw=none,fill=none,text=gray] at ($(d')+(.5,0)$) {\rotatebox{90}{center}};

\draw
(g') edge[comp] (c')
(g) edge[comp] (c)
(g) edge[comp] (c')
(g') edge[comp] (c)
(g''') edge[comp] (c')
(g''') edge[comp] (c)

(w') edge[comp,bend right=35] (x')
(w') edge[   ,bend right=35] (x')
(v') edge[] (w')
(v') edge[] (x')
(w') edge[] (g')
(v') edge[mg] (w')
(v') edge[opt] (x')
(w') edge[comp] (g')
(w') edge[compopt] (g')
(g') edge[] (c)
(g') edge[] (c')

(w) edge[comp,bend right=35] (x)
(w) edge[   ,bend right=35] (x)
(v) edge[] (w)
(v) edge[] (x)
(w) edge[] (g)
(v) edge[mg] (w)
(v) edge[opt] (x)
(w) edge[comp] (g)
(w) edge[compopt] (g)
(g) edge[] (c)
(g) edge[] (c')

(w''') edge[comp,bend right=35] (x''')
(w''') edge[   ,bend right=35] (x''')
(v''') edge[] (w''')
(v''') edge[] (x''')
(w''') edge[] (g''')
(v''') edge[mg] (w''')
(v''') edge[opt] (x''')
(w''') edge[comp] (g''')
(w''') edge[compopt] (g''')
(g''') edge[] (c)
(g''') edge[] (c')

(c) edge[comp] (c'')
(c) edge[mg] (c'')
(c) edge[comp] (d)
(c) edge[compopt] (d)
(d) edge[comp] (c')
(d) edge[] (c')
(c') edge[comp] (d'')
(c') edge[] (d'')
(c'') edge[comp] (d'')
(c'') edge[compopt] (d'')
(c) edge[comp] (d')
(c) edge[] (d')
(c') edge[comp] (d')
(c') edge[compopt] (d')
;


\draw[dash pattern=on 2.5pt off 2.5pt,rounded corners=10pt] (w') -- ($ (w)+(.5,-.2) $) -- (v''');

\end{tikzpicture}

%% file: inapproxGreedyHyperGraphs.tex
\section{Hypergraph Matching}
\label{section_hypergraph_matching}
We study the limitations of greedy algorithms for the more general $
k $-Hypergraph Matching Problem.
In a~$k$-hypergraph an edge may have up to $ k $ nodes.
The goal is to find a maximum set of node disjoint edges.
As for common graphs, a $ \frac{1}{k} $-approximation is easily obtained by greedily picking edges~\cite{kh78}.
We show that greedy adaptive priority algorithms in the vertex model cannot
surpass this trivial worst case guarantee.

$ k $-hypergraph matching is $\NP$-complete:
3-dimensional matching, where each edge has exactly three nodes and the graph is
tripartite, as well as the unrestricted hypergraph matching problem, also called
the set packing problem, belong to Karp's 21 $\NP$-complete problems.
For an overview of problems closely related to hypergraph matching, see Chan and
Lau~\cite{chanLau}.

We consider $ k $-uniform hypergraphs where each edge has exactly $ k $ nodes.
To achieve non-trivial approximation guarantees efficiently, local search was
shown to be successful.
Hurkens and Schrijver~\cite{Hurkens:1989:SSS:63905.63913} gave, for any fixed $ \varepsilon>0 $, a
polynomial time local search algorithm with approximation ratio $
\frac{k}{2}+\varepsilon $.
Using an enhanced local search method, Cygan~\cite{10.1109/FOCS.2013.61}
recently improved the approximation ratio to $ \frac{k+1+\varepsilon}{3} $.
On the other hand, Hazan, Safra, and Schwartz~\cite{Hazan03onthe} showed that $
k $-uniform hypergraph matching cannot efficiently be approximated within a factor
of $ O(\frac{k}{\ln k}) $.

Greedy approaches have also been investigated.
Bennett and Bohman~\cite{1210.3581} showed the following bound on the
expected performance of \greedy on $ k $-uniform $ D $-regular hypergraphs $ H $
with $ N $ nodes:
If $ D\to\infty $ as $ N\to\infty $ and co-degrees are at most $ L = o(D/ \log^5
N) $, then a proportion of at most $ (L/D)^{\frac{1}{2(k-1)} + o(1)} $ of the
nodes remains unmatched whp.
Aronson et al. \cite{adfs95} investigated \greedy on general $ k $-uniform
hypergraphs and showed that the expected approximation ratio is at least $
1/(k-\frac{k-1}{m}) $,
where the non-negative value of $ m $ depends on the graph.
We give a tight bound for greedy adaptive priority algorithms.

\begin{theorem}
No greedy adaptive priority algorithm in the vertex model has approximation ratio better than~$
\frac{1}{k} $ for $ k $-uniform hypergraph matching with~$k \geq 3$.
\end{theorem}
\begin{proof}
Given an algorithm $ A $, we construct a $ k $-uniform hypergraph on which $ A $ has approximation ratio exactly $ \frac{1}{k} $.
The instance constructed by the adversary is illustrated in
Fig.~\ref{kdimhardinstance}.
The white (vertical) edges constitute a maximum matching $ \{e_0,\dots,e_{k-1}\}
$. The topmost horizontal edge~$ e $ will be the only edge picked by
the adaptive priority algorithm.
(We call a node an \emph{$ e $-node} if it belongs to the edge labeled $ e $,
and a \emph{non-$ e $-node} otherwise.)

The adversary creates $ k-1 $ additional edges that are depicted as gray edges
in Fig.~\ref{kdimhardinstance}.
For $ 0\leq i\leq k-2 $ an edge contains the (unique) $ e $-node of vertical edge~$ e_i
$ and a non-$ e $-node of
each $ e_j\neq e_i $ chosen in a way such that all non-$ e $-nodes are covered
at most once.
Exactly one non-$ e $-node of each of $e_0,\dots,e_{k-2} $ is not contained in
a gray edge, call them~$ v_0,\dots,v_{k-2} $.

\begin{figure}[htbp!]
\centering
\begin{minipage}[t]{0.4\textwidth}
\centering
\scalebox{.5}{
\begin{tikzpicture}[baseline=(current bounding box.north)]
\def\w{.2}
\def\c{4}
\tikzstyle{every node} = [circle, fill=white,draw=black,inner sep=0,minimum size=5];
\node                (v1k) {};
\node[above of=v1k]  (v1k-1) {};
\node[above of=v1k-1,draw=none] (dots) {$ \vdots $};

\node[above of=dots] (v13) {};
\node[above of=v13] (v12) {};
\node[above of=v12] (v11) {};

\node[draw=none,above of=v11,font=\LARGE] (e1) {\rotatebox{90}{$ e_{0} $}};

\node[right of=v11]  (v21) {};
\node[below of=v21]  (v22) {};
\node[below of=v22]  (v23) {};
\node[below of=v23,draw=none] (dots) {$\vdots$};
\node[below of=dots]  (v2k-1) {};
\node[below of=v2k-1] (v2k) {};

\node[draw=none,above of=v21,font=\LARGE] (e2) {\rotatebox{90}{$ e_{1} $}};

\node[right of=v21]  (v31) {};
\node[below of=v31]  (v32) {};
\node[below of=v32]  (v33) {};
\node[below of=v33,draw=none] (dots) {$\vdots$};
\node[below of=dots]  (v3k-1) {};
\node[below of=v3k-1] (v3k) {};

\node[draw=none,above of=v31,font=\LARGE] (e2) {\rotatebox{90}{$ e_{3} $}};

\node[draw=none,right of=v3k] (dots) {$ \dots $};
\node[draw=none,right of=v31] (dots) {$ \dots $};

\node[right of=dots] (vk-21) {};
\node[below of=vk-21]  (vk-22) {};
\node[below of=vk-22]  (vk-23) {};
\node[below of=vk-23,draw=none] (dots) {$\vdots$};
\node[below of=dots] (vk-2k-1) {};
\node[below of=vk-2k-1] (vk-2k) {};

\node[draw=none,above of=vk-21,font=\LARGE] (ek-2) {\rotatebox{90}{$ e_{k-3} $}};

\node[right of=vk-21] (vk-11) {};
\node[below of=vk-11]  (vk-12) {};
\node[below of=vk-12]  (vk-13) {};
\node[below of=vk-13,draw=none] (dots) {$\vdots$};
\node[below of=dots] (vk-1k-1) {};
\node[below of=vk-1k-1] (vk-1k) {};

\node[draw=none,above of=vk-11,font=\LARGE] (ek-1) {\rotatebox{90}{$ e_{k-2} $}};

\node[right of=vk-11] (vk1) {};
\node[below of=vk1]  (vk2) {};
\node[below of=vk2]  (vk3) {};
\node[below of=vk3,draw=none] (dots) {$\vdots$};
\node[below of=dots] (vkk-1) {};
\node[below of=vkk-1] (vkk) {};

\node[draw=none,above of=vk1,font=\LARGE] (ek) {\rotatebox{90}{$ e_{k-1} $}};

\draw[rounded corners=\c] ($(v11)+(-\w,\w)$) rectangle ($(v1k)+(+\w,-\w)$);
\draw[rounded corners=\c] ($(v21)+(-\w,\w)$) rectangle ($(v2k)+(+\w,-\w)$);
\draw[rounded corners=\c] ($(v31)+(-\w,\w)$) rectangle ($(v3k)+(+\w,-\w)$);
\draw[rounded corners=\c] ($(vk1)+(-\w,\w)$) rectangle ($(vkk)+(+\w,-\w)$);
\draw[rounded corners=\c] ($(vk-21)+(-\w,\w)$) rectangle ($(vk-2k)+(+\w,-\w)$);
\draw[rounded corners=\c] ($(vk-11)+(-\w,\w)$) rectangle ($(vk-1k)+(+\w,-\w)$);

\draw[rounded corners=\c] ($(v11)+(-\w,\w)$) rectangle ($(vk1)+(+\w,-\w)$);

\node[draw=none,left of=v11,font=\LARGE] (e) {$ e $};

\node[draw=none,left of=v12] {$ \dots $};
\node[draw=none,left of=v13] {$ \dots $};
\node[draw=none,left of=v1k-1] {$ \dots $};
\node[draw=none,right of=vk3] {$ \dots $};
\node[draw=none] at ($(vk3)+(1,-1)$) {$ \dots $};
\node[draw=none,right of=vkk-1] {$ \dots $};
\node[draw=none,right of=vkk] {$ \dots $};

\draw[rounded corners=\c,fill=gray, fill opacity=0.5]
($(v11)+(0,\w)$) --
($(v11)+(.5*\w,\w)$) --
($(v22)+(-2*\w,\w)$) --
($(v22)+(-\w,\w)$) --
($(vk2)+(\w,\w)$) --
($(vk2)+(\w,-\w)$) --
($(v22)+(-\w,-\w)$) --
($(v22)+(-3*\w,-\w)$) --
($(v22)+(-3*\w,+\w)$) --
($(v11)+(-\w,-.25*\w)$) --
($(v11)+(-\w,\w)$) --
($(v11)+(0,\w)$) 

;

\draw[rounded corners=\c,fill=gray!85, fill opacity=0.5]
($(vk3)+(1+2*\w,-\w)$) --
($(v33)+(-3*\w,-\w)$) --
($(v32)+(-3*\w,+\w)$) --
($(v21)+(-\w,-\w)$) --
($(v21)+(-\w,\w)$) --
($(v21)+(0,\w)$) --
($(v21)+(.5*\w,\w)$) --
($(v32)+(-2*\w,+\w)$) --
($(v33)+(-2*\w,+\w)$) --
($(vk3)+(1+2*\w,\w)$)

($(v12)+(-1-2*\w,\w)$) --
($(v12)+(\w,\w)$) --
($(v12)+(\w,-\w)$) --
($(v12)+(-1-2*\w,-\w)$)

;

\draw[rounded corners=\c,fill=gray!70, fill opacity=0.5]
($(vk3)+(1+2*\w,-1-\w)$) --
($(v33)+(1-3*\w,-1-\w)$) --
($(v32)+(1-3*\w,+\w)$) --
($(v31)+(-\w,-\w)$) --
($(v31)+(-\w,\w)$) --
($(v31)+(0,\w)$) --
($(v31)+(.5*\w,\w)$) --
($(v32)+(1-2*\w,+\w)$) --
($(v33)+(1-2*\w,-1+\w)$) --
($(vk3)+(1+2*\w,-1+\w)$)

($(v13)+(-1-2*\w,\w)$) --
($(v23)+(\w,\w)$) --
($(v23)+(\w,-\w)$) --
($(v13)+(-1-2*\w,-\w)$)

;

\draw[rounded corners=\c,fill=gray!55, fill opacity=0.5]

($(vkk-1)+(1+2*\w,\w)$) --
($(vk-1k-1)+(-2*\w,\w)$) --
($(vk-12)+(-2*\w,\w)$) --
($(vk-21)+(.5*\w,\w)$) --
($(vk-21)+(-\w,\w)$) --
($(vk-21)+(-\w,-.25*\w)$) --
($(vk-12)+(-3*\w,\w)$) --
($(vk-1k-1)+(-3*\w,-\w)$) --
($(vkk-1)+(1+2*\w,-\w)$) 

;

\draw[rounded corners=\c,fill=gray!40, fill opacity=0.5]
($(v1k-1)+(-1-2*\w,\w)$) --
($(vk-2k-1)+(+\w,+\w)$) --
($(vk-2k-1)+(+\w,-\w)$) --
($(v1k-1)+(-1-2*\w,-\w)$)

($(vkk)+(1+2*\w,\w)$) --
($(vkk)+(-2*\w,\w)$) --
($(vk2)+(-2*\w,\w)$) --
($(vk-11)+(.5*\w,\w)$) --
($(vk-11)+(-\w,\w)$) --
($(vk-11)+(-\w,-.25*\w)$) --
($(vk2)+(-3*\w,\w)$) --
($(vkk)+(-3*\w,-\w)$) --
($(vkk)+(1+2*\w,-\w)$) 
;

\draw[rounded corners=\c,very thick]

(v1k) -- ($(v1k)+(.5,.5)$) -- ($(v21)+(-.5,-.5)$) -- (v21)

(v2k) -- ($(v2k)+(.5,.5)$) -- ($(v31)+(-.5,-.5)$) -- (v31)

(v3k) -- ($(v3k)+(.5,.5)$) -- ($(v3k)+(.5,1.5)$)

($(vk-21)+(-.5,-3.5)$) -- ($(vk-21)+(-.5,-.5)$) -- (vk-21)

(vk-2k) -- ($(vk-2k)+(.5,.5)$) -- ($(vk-11)+(-.5,-.5)$) -- (vk-11)

(vk-1k) -- ($(vk-1k)+(-.5,-.5)$) -- ($(v1k)+(-.5,-.5)$) -- ($(v11)+(-.5,-.5)$) -- (v11)
;
\end{tikzpicture}
}
\caption{
A hard~$ k $-uniform hypergraph matching instance
}
\label{kdimhardinstance}
\end{minipage}%
\hfill
\begin{minipage}[t]{0.5\textwidth}
\centering
{\footnotesize
\newlength{\colsep}
\setlength{\colsep}{.5em}
\newlength{\rowsep}
\setlength{\rowsep}{.35em}
\begin{alignat*}{2}
\begin{array}{c@{\hspace{\colsep}}c@{\hspace{\colsep}}c@{\hspace{\colsep}}c@{\hspace{\colsep}}c@{\hspace{\colsep}}c@{\hspace{\colsep}}c@{\hspace{\colsep}}c@{\hspace{\colsep}}c@{\hspace{\colsep}}c@{\hspace{\colsep}}c@{\hspace{\colsep}}}
S_0    &=&\{& 1,    & 2,   & 3,    & 4,   &\dots,   & k-2    &\}\\[\rowsep]
S_1    &=&\{& 1,    & k-1,   &k,       &k+1,    &\dots,     &2k-5      &\}\\[\rowsep]
S_2    &=&\{& 2,    &k-1,    & 2k-4, & 2k-3,& \dots,  & 3k-9   &\}\\[\rowsep]
S_3    &=&\{& 3,    &k     , & 2k-4, &        &           &          &\}\\[\rowsep]
       & &  & \vdots&\vdots  & \vdots&        &\ddots     &          &\\[\rowsep]
S_{k-2}&=&\{& k-2,  &2k-5,   & 3k-9, &        &           &         K&\}
\end{array}
\end{alignat*}
}
\caption{Definition of the $ S_i $}
\label{defsi}
\end{minipage}
\end{figure}
So far, the $ e $-nodes of $ e_0,\dots,e_{k-2} $ have degree three, the $ e $-node of $ e_{k-1} $ has degree two.
%
The adversary creates $k-1 $ more edges, that are displayed as black (vertical) lines
in Fig.~\ref{kdimhardinstance}.
These edges use $ K=\frac{(k-1)(k-2)}{2} $ new nodes~$1, 2,\dots,K$. 
Using these nodes, the adversary creates sets $ S_0,\dots,S_{k-2}$
of~$k-2$ nodes each, such that every new node occurs in exactly two of the the $ S_i $ and $
|S_i\cap S_j|=1 $ whenever $ i\neq j $. Refer to Fig.~\ref{defsi} for the
construction of the $ S_0,\dots,S_{k-2}$:
a node listed in
the $ j $-th row has its second occurrence in the $ j $-th column.
The~$i$-th new edge, with $ 0\leq i\leq k-2 $, contains $ v_i $, the $ e
$-node of $ e_{i+1\mod k-1} $, and the nodes of $ S_i $.
%

\newpage

Observe the following properties of the construction:
\begin{enumerate}[label=\roman*.]
  \item\label{eins} The $ e $-nodes of $ e_0,\dots,e_{k-2} $ have degree four.
  \item\label{zwei} All other nodes, including the new nodes in~$S_0,\ldots,S_{k-2}$,
  have degree two.
  \item\label{drei} Any two edges have at most one node in common.
  \item\label{vier} The edge $ e $ shares exactly one node with any other edge.
\end{enumerate}

%
Now the data items of the input graph look as follows:
The data item~$$ \langle u; V_1,\dots,V_d \rangle$$ of node $ u $ lists the~$d$
hyperedges incident in~$ u $: each hyperedge~$ \{u\}\cup V_i $ is represented by
the node set~$V_i$.

How does the greedy adaptive priority algorithm~$A$ proceed when the adaptive
priority game starts?
Recall that $ A $ submits an ordering~$\pi$ on the set of all data items
\emph{without looking at the graph}.
In the first round the adversary presents the, according to $ \pi $, first data
item $ \langle u;V_1,\dots,V_d\rangle $  with $ d \in\{2,4\} $ (which are the
only degrees present in the graph, by~\ref{eins} and~\ref{zwei}), $V_i\cap V_j =
\emptyset$ for $ i\neq j $ (node $ u $ is the only common node of all incident
edges, by~\ref{drei}) and $ |V_i|=k-1 $ for all $ i $ (the graph is $ k
$-uniform).
Since $ A $ is greedy, $ A $ selects an incident edge~$ \{u\}\cup V_i $ and adds
it to its matching.

First assume that~$d = 4$ holds. 
Then the adversary may relabel the nodes in the instance such that $ u $ is the $ e $-node of $ e_0 $, since this is the first data item revealed to
the algorithm.
The greedy adaptive priority algorithm must pick an edge incident to~$u$, and
the adversary asserts that this edge is $ e $.
The matching is maximal by~\ref{vier} 
In case $ d=2 $ the adversary relabels the nodes such that $ u $ is the $ e $-node of $ e_{k-1} $, which has degree two by construction, and again lets the picked edge be $ e $.
\end{proof}

%% file: conclusion.tex
\section{Conclusion}
%
%
Our inapproximability result for fully randomized priority algorithms implies
that greedy-like algorithms cannot compete with algorithms based on
augmenting-paths or algebraic methods.
Nonetheless, conceptually simple algorithms, that are easy to
implement and very efficient in practice, deserve further investigation.

Theorem~\ref{apvUpperBounds} gives inapproximability bounds for a large class of
deterministic greedy algorithms on graphs with maximum degree~$\maxdeg \geq 3$.
We conjecture that the deterministic variant of \MDeg achieves these bounds
for all~$\maxdeg$.

Moreover, our approximation guarantee given in Theorem~\ref{thm:threeDeb} does
not take into account that choosing a random neighbor has a good probability of
picking an optimal neighbor, if the degrees are small (cp.~\cite{ef65} and also
Sect.~\ref{section_intro}).
We leave it as future work to exploit this observation.
%
%

%% file: appendixMinGreedyLinTime.tex
\section{A Linear Time Implementation of \mingreedy}
\label{section_mingreedy_implementation}
For \AlgFrieze and \ranking Poloczek and Szegedy~\cite{ps12} propose a data
structure that allows to run both algorithms in linear time.
\greedy can also be implemented in linear time using a data structure similar
to the one we describe below.

Given an adjacency list representation of the input graph~$G = (V=\{0,\dots,n-1\},E)$, the data
structure can be initialized in linear time~$O(|E| + |V|)$.
At any time during \mingreedy, the data structure supports each of the following
operations in constant time:
selection of a random node of minimum (non-zero) degree, selection of a random
neighbor of a given node, and the deletion of a given edge.
Hence \mingreedy can be implemented in linear time since a minimum degree node
and a neighbor are selected at most $\frac{|V|}{2}$ times and each of the $ |E|
$ edges is removed exactly once.

How is a minimum degree node $ u $ selected in constant time?
Consider a step of \mingreedy and let $ d_0<d_1<\dots<d_k $ be the different
degrees currently present in the graph, where $ d_0=0 $ is the degree of already isolated nodes.
We use an array $ S $ which is partitioned into sub-arrays $ S_i $
($ 0\leq i\leq k $) such that~$S_i$ precedes all~$S_j$ with~$i < j$.
Each~$S_i$ contains all nodes that currently have degree $ d_i $ in contiguous cells of $ S $.
A doubly linked list $ D $ stores (from head to tail) the borders of $
S_0,S_1,\dots,S_k $.
Node $ u $ is selected by reading the second entry in $ D $, which stores nodes of minimum degree $ d_1 $, and choosing a random
cell in $ S_1 $.

To select a random neighbor $ v $ of $ u $ in constant time, instead of adjacency lists we utilize adjacency \emph{arrays} (which have same lengths as the lists and can be computed in linear time during the preprocessing phase).
To pick $ v $ from the
adjacency array $ A_u $ of $ u $, we assert that the currently available
neighbors of~$u$ are stored in a consecutive part of~$ A_u $.

Now that nodes~$u$ and~$v$ are selected, the edge~$(u,v)$, and all incident edges
of $u$ and~$v $ are removed from the data structure. We remove these edges
one-by-one, each in constant time.

This is how we update the adjacency arrays.
Let a node $ x $ and an index of a cell in $ A_x $ be given, say containing
neighbor $ y $.
Assume that the array cell in $ A_x
$ containing node $ y $ also stores the index of the array cell in $ A_y $
containing node $ x $ as a \emph{reference} and vice versa.
In order to remove the edge $ (x,y) $, we move the entry of the last non-empty
cell in~$A_x$ to the position of $ y $, and handle~$A_y $ and $ x $ analogously.
We also update the references of the two moved entries accordingly;
this is done in constant time using the references stored inside the moved entries.
The references are initialized during the construction of the adjacency arrays.

How to update $ S $ and $ D $ for the removal of an edge $ (x,y) $?
Since the degrees $ d_x,d_y $ of $ x$ respectively $y $ are decreased by exactly one, these nodes are moved from sub-array $ S_{d_x} $ to sub-array $ S_{d_x-1} $ respectively from $ S_{d_y} $ to $ S_{d_y-1} $.
We proceed analogously for $ x $ and $ y $.
To move $ x $ to its new sub-array, we utilize two helper arrays $
P_D,P_S $:
$ P_D[x]$ stores a pointer to the $ D $-entry containing $ x $, $
P_S[x] $ holds the index of the cell in $ S $ containing $ x $.
The entry of node $ x $ in the sub-array $ S_{d_x} $ is replaced by the ``leftmost'' node
in $ S_{d_x} $, i.e., the node stored at the smallest index, and $ x $ is appended to the ``right'' of $ S_{d_x-1} $.
If $ S_{d_x} $ is now empty, we remove it from $ D $ in constant time using the pointer $ P_D[x] $.
If $ S_{d_x-1} $ does not yet exist, i.e., $ x $ is now the only node of degree $ d_x-1 $, then we create $ S_{d_x-1} $ at the now cleared position in $ S $ and insert $ S_{d_x-1} $ before $ S_{d_x} $ (or before the successor of $ S_{d_x} $, if $ S_{d_x} $ was removed), also in constant time.
The pointers in $ P_D $ and the addresses stored in $ P_S $ are updated accordingly.

Note that $ S,D,P_S,P_D $ can be initialized in time $ O(|E|+|V|) $ during the preprocessing phase by scanning through the adjacency lists of $ G $ a constant number of times.